\documentclass[twocolumn,table,dvipsnames]{article}
\usepackage{fullpage}
\usepackage{setspace}
\usepackage{graphicx}
\usepackage{amssymb}
\usepackage{amsmath}
\usepackage{amsthm}
\usepackage{graphicx}
\usepackage{latexsym}
\usepackage{epstopdf}
\usepackage{multicol}
\usepackage{url}
\usepackage{fancyvrb}
\usepackage{subcaption}
\usepackage{paralist}
\usepackage{listings}
\usepackage{mathtools}
\usepackage[skip=0.5pt]{caption}
\usepackage[vlined,ruled,commentsnumbered]{algorithm2e}
\usepackage{tikz,tkz-graph}
\usepackage{pifont}
\usepackage{standalone}
\usepackage{multirow}

\usetikzlibrary{trees,topaths,calc,shapes.geometric,arrows.meta,positioning,intersections,fit,decorations.pathreplacing}
\newcommand{\cmark}{\ding{51}}%
\newcommand{\xmark}{\ding{55}}%
\colorlet{red1}{purple!40!gray!60}

\newcommand{\specialcell}[2][c]{%
  \begin{tabular}[#1]{@{}l@{}}#2\end{tabular}}

\colorlet{circle edge}{blue!50}
\colorlet{circle area}{blue!20}

\tikzset{filled/.style={fill=circle area, draw=circle edge, thick},
    outline/.style={draw=circle edge, thick}}

\tikzset{
    every node/.style={font=\sffamily\small},
    main node/.style={thick,circle,draw,font=\sffamily\Large}
}

\definecolor{keywords}{rgb}{0,0,0.7}

\newcommand{\cut}[1]{}

\newcommand{\reminderc}[1]{\textcolor{blue}{[[[#1]]]}}

\newcommand{\sudeepa}[1]{{\textcolor{red}{\reminderc{\bf (Sudeepa)~#1}{\typeout{#1}}}}}

\newcommand{\amir}[1]{}

\newcommand{\reva}[1]{\color{black}{#1}}
\newcommand{\revb}[1]{\color{black}{#1}}
\newcommand{\revc}[1]{\color{black}{#1}}
\newcommand{\common}[1]{\color{black}{#1}}

\newcommand{\attrset}{\mathbb{A}}
\newcommand{\dom}{{\tt dom}}
\newcommand{\holo}{{\tt HoloClean}}

\newcommand{\ad}{{\tt ag}}
\newcommand{\cc}{{\tt c}}
\newcommand{\au}{{\tt a}}
\newcommand{\ww}{{\tt w}}

\newcommand{\pp}{{\tt p}}
\newcommand{\f}{{\tt g}}

\newcommand{\auth}{{\tt Author}}
\newcommand{\adv}{{\tt AuthGrant}}
\newcommand{\fnd}{{\tt Grant}}

\newcommand{\cites}{{\tt Cite}}
\newcommand{\writes}{{\tt Writes}}
\newcommand{\pub}{{\tt Pub}}
\newcommand{\dl}{~{:}{-}~}
\renewcommand{\vec}[1]{{\mathbf{#1}}}

\newtheorem{theorem}{Theorem}[section]
\newtheorem{proposition}[theorem]{Proposition}
\newtheorem{example}[theorem]{Example}

\newtheorem{definition}[theorem]{Definition}

\newtheorem*{lemma*}{Lemma}
\newtheorem{corollary}[theorem]{Corollary}

\AtBeginDocument{%
  \providecommand\BibTeX{{%
    \normalfont B\kern-0.5em{\scshape i\kern-0.25em b}\kern-0.8em\TeX}}}

\title{On Multiple Semantics for Declarative Database Repairs}
\author{
  Amir Gilad\\ Tel Aviv University \\ amirgilad@mail.tau.ac.il
  \and
  Daniel Deutch \\ Tel Aviv University \\ danielde@post.tau.ac.il
  \and
  Sudeepa Roy \\ Duke University \\ sudeepa@cs.duke.edu
}
\date{}
\begin{document}

\maketitle

\begin{abstract}
We study the problem of Database Repairs through a rule-based framework that we refer to as {\em Delta Rules}. Delta Rules are highly expressive and allow specifying complex, cross-relations repair logic associated with Denial Constraints, Causal Rules, and allowing to capture Database triggers of interest. We show that there are no one-size-fits-all semantics for repairs in this inclusive setting, and we consequently introduce multiple alternative semantics, presenting the case for using each of them. We then study the relationships between the semantics in terms of their output and the complexity of computation. Our results formally establish the tradeoff between the permissiveness of the semantics and its computational complexity. 
We  demonstrate the usefulness of the framework in capturing multiple data repair scenarios for an Academic Search database and the TPC-H databases, showing how using different semantics affects the repair in terms of size and runtime, and examining the relationships between the repairs. We also compare our approach with SQL triggers and a state-of-the-art data repair system.
\end{abstract}

\section{Introduction}\label{sec:intro}


The problem of data repair has been extensively studied by previous work \cite{RekatsinasCIR17,ChuIP13,ChomickiM05,BertossiKL13,FaginKK15,Afrati2009}. Many of these have focused on the desideratum of minimum cardinality, i.e., repairing the database while making the minimum number of changes \cite{Afrati2009,FaginKK15,LopatenkoB07}. In particular, when the repair only involves tuple deletion \cite{ChomickiM05,LopatenkoB07,LivshitsKR18}, this desideratum takes center stage since a n\"aive repair could simply delete the entire database in order to repair it. Such repairs are commonly used with classes of constraints such as Denial Constraints (DCs) \cite{ChomickiM05,ChuIP13}, SQL deletion triggers \cite{DBSystems0020812}, and causal dependencies \cite{RoyS14}. 

Different scenarios, however, may require different interpretations of the constraints and the manner in which they should be used to achieve a minimum repair.
For integrity constraints such as DCs,
when there is a set of tuples violating such a DC, 
any tuple in that set is a `\emph{candidate for deletion}' to repair the database. Moreover, if we allow such constraints to be influenced by deleted tuples, as needed in cascade deletions, the problem becomes more convoluted. 

In contrast, for violations of referential integrity constraints under cascade delete semantics, or other complex and user-defined constraints as in SQL triggers 
and in causal dependencies,
there is a specific tuple that is meant to be deleted if a trigger or a rule is satisfied. 
Nevertheless, if there are several triggers or causal rules, all satisfied at the same time, it remains largely unspecified and varies from system to system in what order they should be fired and when should the database be updated. For instance, by default, MySQL chooses to fire triggers in the order they have been created \cite{mysql}, and PostgreSQL fires triggers in alphabetical order in such scenarios \cite{postgres}. This may lead to different answers leaving users uncertain about why the tuples have been deleted. These systems offer an option of specifying the order in which the triggers would fire; however, this order does not guarantee a consistent semantics that leads to a minimum repair. Such constraints may also follow several different semantics in the process of cascading deletions. 
Therefore, the same set of constraints may be assigned different reasonable semantics that lead to different minimum repairs, and each choice of semantics may be suitable for a different setting. 





\begin{example}\label{ex:first}
\amir{R2,D2,D10: revised example:}
{\revb
Consider the database in Figure \ref{fig:dblp} based on an academic database \cite{mas}. It contains the tables \fnd\ (grant foundations), \auth\ (paper authors), \adv\ (a relationship of authors and grants given by a foundation), \pub\ (a publication table), \writes\ (a relationship table between \auth\ and \pub  ), and \cites\
(a citation table of citing and cited relationships). For each tuple, we also have an identifier on the leftmost column of each table (e.g., $\ad_1$ is the identifier of $\mathtt{\adv(2,1)}$).
Consider the following four constraints specifying how to repair the tables (there could be other rules capturing different repair scenarios): 
\begin{compactenum}
    \item If a \fnd\ tuple is deleted and there is an author who won a grant by this foundation, denoted as an \adv\ tuple, then delete the winning author.
    \item If an \auth\ tuple is deleted and the corresponding \writes\ and \pub\ tuples exist in the database, delete the corresponding \writes\ tuple (as in cascade delete semantics for foreign keys).
    \item 
    Under the same condition as above, delete the corresponding \pub\ tuple (not standard foreign keys, but suggesting that every author is important for a publication to exist).
    \item If a publication $p$ from the \pub\ table is deleted, and is cited by another publication $c$, while some authors of these papers still exist in the database, then delete the \cites\ tuple\footnote{{\revb An alternative version of this constraint is not conditioned by the existence of the paper authors, however, the condition is added to demonstrate a difference between the semantics in our framework}}.
\end{compactenum}
Suppose we are analyzing a subset of this database 
containing only authors affiliated with U.S. schools and only papers written solely by U.S. authors.
ERC grants are given only to European institutions and its \fnd\ tuple was incorrectly added to the U.S. database, so this tuple needs to be deleted.
However, this deletion causes violations in the above constraints.
To repair the database based on these constraints, we could proceed in various ways: 
considering the semantics of triggers and causal rules, we can delete tuples $\au_2$, $\ww_1$, $\pp_1$, $\au_3$, $\ww_2$, $\pp_2$ and $\cc$, and regain the integrity of the database but at the cost of deleting seven tuples. A different approach is to delete $\au_2$ and either $\ww_1$ or $\pp_1$, and delete $\au_3$ and either $\ww_2$ or $\pp_2$, which would only delete four tuples. 
However, if we consider the semantics of DCs, we could delete any tuple out of the set of tuples that violates the constraint. So, we can just delete the tuples $\ad_2, \ad_3$. This would satisfy the first constraint and thus the second, third and fourth constraints will also be satisfied. 
}
\end{example}

\subsection*{Our  Contributions} In this paper, we propose a novel unified constraint specification framework, with multiple alternative semantics that can be suitable for different settings, and thus can result in different `minimum repairs'. Our framework allows for semantics similar to DCs as well as causal rules, and the subset of SQL triggers that delete tuple(s) after another deletion event, 
and is geared toward minimum database repair using tuple deletions.

\cut{
We determine the relationships between the minimum repairs under the different semantics, and analyze the complexity of finding the minimum repair under each semantic. 
For the intractable cases, we find efficient algorithms that  work well in practice. Finally, we experimentally examine the performance of our algorithms and compare the repairs under the different semantics. 
We now detail the contributions of this work.
}

\begin{figure}[!htb]
    \centering \tiny
    \begin{minipage}{.35\linewidth}
        \centering
        \caption*{$\fnd$}\label{tbl:advised}
        \begin{tabular}{ c | c | c | c | c | c |}
          \cline{2-3}  & gid & name \\
            \hline $\f_1$ & 1 & NSF \\
            \hline $\f_2$ & 2 & ERC \\
            \hline
        \end{tabular}
    \end{minipage}%
    \begin{minipage}{.3\linewidth}
        \centering
        \caption*{ $\adv$}\label{tbl:authgrant}
        \begin{tabular}{ c | c | c | c | c | c |}
          \cline{2-3}  & aid & gid \\
            \hline $\ad_1$ & 2 & 1 \\
            \hline $\ad_2$ & 4 & 2 \\
            \hline $\ad_3$ & 5 & 2 \\
            \hline
        \end{tabular}
    \end{minipage}%
    \begin{minipage}{.35\linewidth}
        \centering
        \caption*{ $\auth$}\label{tbl:author}
        \begin{tabular}{c | c | c | c | c | c |}
            \cline{2-3} & aid & name \\
            \hline $\au_1$ & 2 & Maggie \\
            \hline $\au_2$ & 4 & Marge \\
            \hline $\au_3$ & 5 & Homer \\
            \hline
        \end{tabular}
    \end{minipage}
    
    \begin{minipage}{.33\linewidth}
        \centering
        \caption*{$\cites$}\label{tbl:cite}
        \begin{tabular}{c | c | c | c | c | c |}
           \cline{2-3}  & citing & cited\\
            \hline $\cc$ & 7 & 6 \\
            \hline
        \end{tabular}
    \end{minipage}%
    \begin{minipage}{.33\linewidth}
        \centering
        \caption*{$\writes$}\label{tbl:writes}
        \begin{tabular}{c | c | c | c | c | c |}
           \cline{2-3}  & aid & pid\\
            \hline $\ww_1$ & 4 & 6 \\
            \hline $\ww_2$ & 5 & 7 \\
            \hline
        \end{tabular}
    \end{minipage}%
     \begin{minipage}{.33\linewidth}
        \centering
        \caption*{$\pub$}\label{tbl:publication}
        \begin{tabular}{c | c | c | c | c | c |}
           \cline{2-3}  & pid & title\\
            \hline $\pp_1$ & 6 & x \\
            \hline $\pp_2$ & 7 & y \\
            \hline
        \end{tabular}
    \end{minipage}
    \caption{{\revb Academic database instance $D$}}\label{fig:dblp}
\end{figure}

\begin{figure}[!htb]
\begin{footnotesize}
 \begin{lstlisting}
(0) $\Delta_{\fnd}(g,n)$ :- $\fnd(g,n)$, $n = `ERC'$
(1) $\Delta_{\auth}(a,n)$ :- $\auth(a,n)$, $\adv(a,g)$, $\Delta_{\fnd}(g,gn)$
(2) $\Delta_{\pub}(p,t)$ :- $\pub(p,t), \writes(a,p), \Delta_{\auth}(a,n)$
(3) $\Delta_{\writes}(a,p)$ :- $\pub(p,t), \writes(a,p), \Delta_{\auth}(a,n)$
(4) $\Delta_{cite}(c,p)$ :- $\cites(c,p), \Delta_{\pub}(p,t), \writes(a_1,c),$ $\writes(a_2,p)$
\end{lstlisting}
\end{footnotesize}
\caption{{\revb Delta program}}\label{fig:deltaprog}
\vspace{-3mm}
\end{figure}



\textbf{Delta rules and stabilizing sets.~}
We begin by defining the concept of delta rules.
Delta rules allow for a deletion of a tuple from the database if certain conditions hold. Intuitively, delta rules are constraints specifying conditions that, if satisfied, compromise the integrity of the database. 
A stabilizing set is a set of tuples whose removal from the database ensures that no delta rules are satisfied.

\begin{example}\label{ex:third}
Reconsider Example \ref{ex:first} where the constraints are specified verbatim. We can formalize them in our declarative syntax, as shown in Figure \ref{fig:deltaprog}. Rules (1), (2), (3), and (4) express the constraints in Example \ref{ex:first}, respectively. For example, rule (3) states that if a \pub\ and a \writes\ tuples exist in the database, and the corresponding \auth\ tuple has been deleted (the $\Delta_{\auth}(a,n)$ atom), then delete the \pub\ tuple (this is the head of the rule). \amir{R2,D10:} {\revb Rule (0) initializes the deletion process (more details about this in Section \ref{sec:deltarules}).}
In Example \ref{ex:first}, $\{\au_2,\au_3, \ww_1, \ww_2,\pp_1,\pp_2, \cc\}$, $\{\au_2,\au_3,\ww_1, \ww_2,\pp_1,\pp_2\}$, $\{\au_2,\au_3,\ww_1, \ww_2\}$, and $\{\ad_2, \ad_3\}$ are all stabilizing sets, as a removal of any of these sets of tuples and an addition of these tuples to the delta relations ensure that no delta rules are satisfied. 
\end{example}

Although we can easily verify that the deletion of any set of tuples in Example~\ref{ex:third} guarantees that the  database is `stable', it may not be immediately obvious under what scenarios we would obtain these sets as the answer, or whether they  correspond to some notion of `optimal repair'. 

\textbf{Semantics of delta rules.~} 
 To address this, we define four semantics of delta rules and define the minimum repair according to these. \amir{R3,D3:} {\revc A semantics in this context implies a manner in which we interpret the rules, either as integrity constraints for which we define a global minimum solution, or as means of deriving tuples in different ways.
 {\em Independent semantics} aims at finding the globally optimum repair such that none of the rules are satisfied on the entire database instance. It is similar to optimal repair in presence of DCs like violations of functional dependencies \cite{ChomickiM05}, but delta rules capture more general propagations of conflict resolutions, where deleting one tuple to resolve a conflict may lead to deletion of another tuple. 
 {\em Step semantics}, is geared towards the semantics of the aforementioned subset of SQL triggers and causal rules, and is a fine-grained semantics. It evaluates one rule at a time (non-deterministically) and updates the database immediately by removing the tuple, which may in turn lead to further tuple deletions. 
 {\em Stage semantics} also aims to capture triggers and causal rules. However, as opposed to step semantics, it deterministically removes tuples in stages. In particular, it evaluates all delta rules based on the stage of the database in the previous round, and therefore the order of firing the rules does not matter (similar to semin\"aive evaluation of datalog \cite{AbiteboulHV95}).
 Finally, {\em end semantics} is similar to the standard datalog evaluation, where all possible delta tuples are first derived and the database is updated at the end of the evaluation process. We use end semantics as a baseline for the other semantics.}

\begin{example}\label{ex:four}
Continuing Example \ref{ex:third}, the results corresponding to different semantics are 
$End(P,D) = \{\au_2,\au_3,\ww_1, \ww_2,\pp_1,\pp_2, \cc\}$, $Stage(P,D) = \{\au_2,\au_3,\ww_1, \ww_2,\pp_1,\pp_2\}$,
$Step(P,D) = \{\au_2,\au_3,\ww_1, \ww_2\}$,
and $Ind(P,D) = \{\ad_2, \ad_3\}$. \amir{R2,D3:} We detail the formal definitions of the semantics in Section \ref{sec:deltarules}.
\end{example}

\textbf{Relationships between the results of different semantics.~}
We study the relationships of containment and size between the results according to the four semantics. The results are summarized in Figure \ref{fig:mss_classes}, where the size of the result of independent semantics is always smaller or equal to the sizes of the results of stage and step semantics. We show there are case where the result of step semantics subsumes the result of stage semantics and vice-versa. 


\textbf{Complexity of finding the results}
We show that finding the result for stage and end semantics is PTIME, while finding the result for step and independent semantics is NP-hard (also shown in Figure \ref{fig:mss_classes}). 
For independent semantics, we devise an efficient algorithm using \emph{data provenance}, leveraging a reduction to the min-ones SAT problem \cite{KratschMW16}. We store the provenance \cite{GreenKT07} as a Boolean formula and find a satisfying assignment that maps the minimum number of negated variables to True. 
For step semantics, we also devise an efficient algorithm based on the structure of the provenance graph, traversing it in topological order and choosing tuples for the result set as we go.  

\textbf{Experimental evaluation.~}
We examine the performance of our algorithms for a variety of programs with varying degree of complexity on an academic dataset \cite{mas} and the TPC-H dataset \cite{tpc}. We measure the performance in terms of subsumption relationship between the results computed under different semantics, the size of these results, the  execution time of each algorithm to compute the result for every semantics. Finally, for our heuristic algorithms, we break down the execution time in the context of multiple ``classes" of programs. We also compare our approach to SQL triggers in PostgreSQL and MySQL, {\revc and to \holo\ \cite{RekatsinasCIR17}}. 

\begin{figure}
\centering
\begin{tikzpicture}[scale=0.5]
 \node (a) [align=center] at (-2,4) {End \\ (PTime)};
 \node (b) [align=center] at (-5,2) {Stage \\ (PTime)};
 \node (c) [align=center] at (1,2) {Step \\ (NP-hard)};
 \node (d) [align=center] at (-2,0) {Independent \\ (NP-hard)};
 \draw (d) -- (b) -- (a);
 \draw (d) -- (c) -- (a);
 \draw[] (b) node[align=center] at (-3.9,3.2)
    {
      \rotatebox[origin=c]{45}{$\subseteq$}
    }(b);
 \draw[->] (b) node[align=center] at (0,3.3)
    {
      \rotatebox[origin=c]{-45}{$\supseteq$}
    }(c);
 \draw[] (c) node[align=center] at (-1,1.3)
    {
      \rotatebox[origin=c]{45}{$\leq$}
    }(d);
 \draw[] (c) node[align=center] at (-2.8,1.3)
    {
      \rotatebox[origin=c]{-45}{$\geq$}
    }(d); 
\end{tikzpicture}
\caption{
Complexity and relationships among the different semantics by size and containment
}
 \vspace{-0.3cm}
\label{fig:mss_classes}
\end{figure}
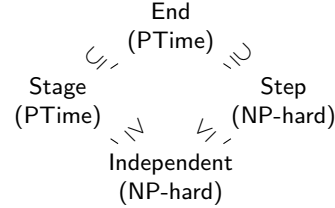

\vspace{-2mm}
\section{preliminaries}\label{sec:prelim}
We start by reviewing basic definitions for databases and non-recursive datalog programs. 
%
A relational schema is a tuple $\mathbf{R} = (R_1, \ldots, R_k)$ where $R_i$ is a relation name (or atom). Each relation $R_i$ ($i = 1$ to $k$) has a set of attributes $\attrset_i$, and we use $\attrset = \cup_i \attrset_i$
to denote the set of all attributes in $\mathbf{R}$. For any attribute $A \in  \attrset$, $\dom(A)$ denotes the domain of $A$. 
A database instance $D$ is a finite set of tuples over $\mathbf{R}$, and we will use $R_1, \cdots, R_k$ to denote both the relation names and their content in  $D$ where it is clear from the context.

 \textbf{Non-recursive datalog.~} We will use standard datalog program comprising rules of the form $Q(\vec{X}) \dl T_{i_1}(\vec{Y}_{i_1}), \ldots, T_{i_\ell}(\vec{Y}_{i_\ell})$, where $\vec{Y_{i_1}}, \cdots, \vec{Y_{i_\ell}}$ contain variables or constants, and $\vec{X}$ is a subset of the variables in $\cup_{j = 1}^{\ell} \vec{Y_{i_j}}$. In this rule,
 $Q$ is called an \emph{intensional} (or derived) relation, and  $T_i$'s 
are either intensional relations or are \emph{extensional} (or base) relations from  
 $\mathbf{R}$.
 For brevity, we use the notations $body(Q)$ for the set  $\{T_{i_1}(\vec{Y}_1), \ldots, T_{i_\ell}(\vec{Y}_{i_\ell})\}$, and $head(Q)$ for $Q(\vec{X})$. A datalog program is simply a set of datalog rules. 
In this paper,  we consider programs $P = \{r_1, \ldots, r_m\}$ such that for some $i, j$, the relation name of $head(r_i)$ is an element in $body(r_j)$, but $P$ is equivalent to a non-recursive program. These are called bounded programs and are not inherently recursive \cite{AbiteboulHV95}. 
 
 Let $D$ be a database and $Q(\vec{X}) \dl T_{i_1}(\vec{Y}_{i_1}), \ldots, T_{i_\ell}(\vec{Y}_{i_\ell})$ be a datalog rule, both over the schema $\mathbf{R}$ (i.e., $\forall R_i \in body(Q).~ R_i \in \mathbf{R}$). An assignment to $Q$ is a function $\alpha: body(Q) \to D$ that respects relation names. We require that a variable $y_j$ will not be mapped to multiple distinct values, and a constant $y_j$ will be mapped to itself. We define $\alpha(head(Q))$ as the tuple obtained from $head(Q)$ by replacing each occurrence of a variable $x_i$ by $\alpha(x_i)$.
 
 Given a database $D$ and a datalog program $P$, we say that it has reached a {\em fixpoint} if no more tuples can be added to the result set using assignments from $D$ to the rules of $P$. The fixpoint, denoted by $P(D)$, is then the database obtained by adding to $D$ all tuples derived from the rules of $P$.
\begin{example}\label{ex:fixpont_end}
Consider the database $D$ in Figure \ref{fig:dblp} and the program $P$ in Figure \ref{fig:deltaprog}, and consider for now all $\Delta$ relations as standard intensional relations. 
When the rules are evaluated over the database, after deriving $\Delta_{\fnd}(2,ERC)$ from rule (0), we have two assignments to rule (1): $\alpha_1$, $\alpha_2$, where $\alpha_1$ ($\alpha_2$) maps the first, second and third atoms to  $\au_2$ ($\au_3$), $\ad_2$ ($\ad_3$), and $\Delta_{\fnd}(2,ERC)$ respectively, which generate $\Delta_{\auth}(4,Marge)$ and $\Delta_{\auth}(5,Homer)$. 
Next we have two assignments to rule (2): $\alpha_3$, $\alpha_4$, where $\alpha_{3}$ ($\alpha_{4}$) maps the first, second and third atoms to $\pp_2$ ($\pp_3$),  $\ww_1$ ($\ww_2$), and $\Delta_{\auth}(4,Marge)$ ($\Delta_{\auth}(5,Homer)$) respectively. The fixpoint of this evaluation process is the database $P(D) = D\cup \{\Delta_{\fnd}(2,ERC), \Delta_{\auth}(4,Marge),$ $\Delta_{\auth}(5,Homer),  \Delta_{\writes}(4,6), \Delta_{\writes}(5,7), \Delta_{\pub}(6,x),$ $\Delta_{\pub}(7,y),$ $\Delta_{\cites}(7,6)\}$. This evaluation corresponds to \emph{end semantics} as discussed later.
\end{example}

\section{Framework For Delta Rules}\label{sec:deltarules}
We now formulate the model used in the rest of the paper. 

\subsection{Delta Relations, Rules, and  Program}
\textbf{Delta Relations.~} Given a schema $\mathbf{R} = (R_1, \ldots, R_k)$ where $R_i$ has attributes $\attrset_i$, the delta relations $\mathbf{\Delta} =  (\Delta_1, \ldots, \Delta_k)$ will be used to capture tuples to be deleted from $R_1, \ldots. R_k$ respectively.  Therefore, each relation $\Delta_i$ has the same set of attributes $\attrset_i$ {\revb (the `full' notation for $\Delta_i$ is $\Delta_{R_i}$, but we  abbreviate it)}. 


\cut{
\begin{definition}
Given a schema $\mathbf{R} = \{R_1, \ldots, R_k\}$, for each $R_i$ with the attribute tuple $(A^i_1, \ldots, A^i_m)$, define the delta relation $\Delta_i$ with the same attribute tuple.
\end{definition}

To shorten notation, we will denote by $\Delta(\mathbf{R})$ the relation set $\{\Delta_1, \ldots, \Delta_k\}$, where $\mathbf{R} = \{R_1, \ldots, R_k\}$. 
}


\textbf{Delta rules and program.~} 
A delta program is a datalog program where every intensional relation is of the form $\Delta_i$ for some $i$.

\begin{definition}\label{def:deltarule}
Given a schema $\mathbf{R} = (R_1, \ldots, R_k)$ and the corresponding delta relations $\mathbf{\Delta} =  (\Delta_1, \ldots, \Delta_k)$, a delta rule is a datalog rule of the form $\Delta_i(\vec{X}) \dl R_i(\vec{X}), Q_1(\vec{Y}_1), \ldots, Q_l(\vec{Y}_l)$ where $Q_i \in \mathbf{R} \cup \mathbf{\Delta}$.
	
\end{definition}

Intuitively, the condition $Q_i \in \mathbf{R} \cup \mathbf{\Delta}$ means that delta rules can have cascaded deletions when some of the other tuples are removed.  
Note that the same vector $\vec{X}$ that appears in the head $\Delta_i$, also appears in the body in the atom with relation $R_i$. This is because we need the atom $R_i(X)$ in the body of the rule so that we only delete existing facts. Also, $Y_i$ can intersect with $X$ and any other $Y_j$. 
We will refer to a set of delta rules as a {\em delta program}.

\begin{example}
Consider rule (2) in Figure \ref{fig:deltaprog}. This rule is meant to delete any $\pub$ tuple after its $\auth$ tuple has been deleted, intuitively saying that if an author of a paper was deleted, then her associated papers should be deleted as well.
We have $\Delta_{\pub}(p,t)$ in the head of the rule and in the body we have the atom $\pub(p,t)$ to make sure the deleted tuple exists in the database and we have a join between the $\pub$ atom and the $\Delta_{\auth}(a,n)$ atom using the atom $\writes(a,p)$. 
\end{example}

Overloading notation, we shall use $\Delta$ also as a mapping from any subset of tuples in $\mathbf{R}$ to $\mathbf{\Delta}$ in the instance $D$. For instance, for two tuples from $R_1, R_2$ as $S = \{R_1(a), R_2(b)\}$, we will use $\Delta(S)$ to denote $\Delta_1(a)$ and $\Delta_2(b)$ suggesting that these two tuples have been deleted.

\subsection{Independent Semantics}\label{sec:indep}
This non-operational `ideal' semantics captures the intuition of a minimum-size repair: the smallest set of tuples that need to be removed so that all the constraints are satisfied. Note that whenever we delete a tuple, we add the corresponding delta tuple. Hence the following definition:




\begin{definition}\label{def:indsem}
Let $D$ and $P$ respectively be a database instance and a delta program over the schema $\mathbf{R},  \Delta$.   The result of independent semantics, denoted $Ind(P,D)$, is the smallest subset of non-delta tuples 
$S\subseteq D$ 
such that in $(D\setminus S) \cup \Delta(S)$ there is no satisfying assignment for any rule of $P$. 
\end{definition}
Note that there may be multiple minimum size sets satisfying the criteria, in which case the independent semantics will non-deterministically output one of them. 
Proposition \ref{prop:basic} shows that there is always a result for this semantics.

\begin{figure}[!htb]
    \centering \tiny
    \begin{minipage}{.35\linewidth}
        \centering
        \caption*{$\fnd$}\label{tbl:grant_changed}
        \begin{tabular}{ c | c | c | c | c | c |}
          \cline{2-3}  & fid & name \\
            \hline $\f_1$ & 1 & NSF \\
            \rowcolor{lightgray}
            \hline $\f_2$ & 2 & ERC \\
            \hline
        \end{tabular}
    \end{minipage}%
    \begin{minipage}{.3\linewidth}
        \centering
        \caption*{ $\adv$}\label{tbl:authgrant_changed}
        \begin{tabular}{ c | c | c | c | c | c |}
          \cline{2-3}  & aid & fid \\
            \hline $\ad_1$ & 2 & 1 \\
            \rowcolor{Cyan}
            \hline $\ad_2$ & 4 & 2 \\
            \rowcolor{Cyan}
            \hline $\ad_3$ & 5 & 2 \\
            \hline
        \end{tabular}
    \end{minipage}%
    \begin{minipage}{.35\linewidth}
        \centering
        \caption*{ $\auth$}\label{tbl:author_changed}
        \begin{tabular}{c | c | c | c | c | c |}
            \cline{2-3} & aid & name \\
            \hline $\au_1$ & 2 & Maggie \\
            \rowcolor{YellowGreen}
            \hline $\au_2$ & 4 & Marge \\
            \rowcolor{YellowGreen}
            \hline $\au_3$ & 5 & Homer \\
            \hline
        \end{tabular}
    \end{minipage}
    
    \begin{minipage}{.33\linewidth}
        \centering
        \caption*{$\cites$}\label{tbl:cite_changed}
        \begin{tabular}{c | c | c | c | c | c |}
           \cline{2-3}  & citing & cited\\
           \rowcolor{Apricot}
            \hline $\cc$ & 7 & 6 \\
            \hline
        \end{tabular}
    \end{minipage}%
    \begin{minipage}{.33\linewidth}
        \centering
        \caption*{$\writes$}\label{tbl:writes_changed}
        \begin{tabular}{c | c | c | c | c | c |}
           \cline{2-3}  & aid & pid\\
            \rowcolor{YellowGreen}
            \hline $\ww_1$ & 4 & 6 \\
            \rowcolor{YellowGreen}
            \hline $\ww_2$ & 5 & 7 \\
            \hline
        \end{tabular}
    \end{minipage}%
     \begin{minipage}{.33\linewidth}
        \centering
        \caption*{$\pub$}\label{tbl:publication_changed}
        \begin{tabular}{c | c | c | c | c | c |}
           \cline{2-3}  & pid & title\\
            \rowcolor{Lavender}
            \hline $\pp_1$ & 6 & x \\
            \rowcolor{Lavender}
            \hline $\pp_2$ & 7 & y \\
            \hline
        \end{tabular}
    \end{minipage}
    \caption{\amir{R1,D3/R2,D5: removed fig. (b) and revised text according to example} {\common The database instance $D$ after applying the rules in Figure \ref{fig:deltaprog} with the different semantics (not showing delta relations). The tuple $\f_2$ is always deleted and added to the delta relation. Tuples of a certain color are deleted from the original relations and added to the delta relations. (1)
    Independent semantics deletes the gray and cyan tuples. (2) Step semantics deletes the gray and green tuples. (3) Stage semantics deletes the gray, green and pink tuples. (4) End semantics deletes the gray, green, pink, and orange tuples and adds them to the delta relations}}\label{fig:indres}
    \vspace{-4mm}
\end{figure}
\begin{example}\label{ex:ind}
Consider the database in Figure \ref{fig:dblp} and the rules shown in Figure \ref{fig:deltaprog}. The result of independent semantics is $\{\f_2,\ad_2,\ad_3\}$ and the final state of the database appears in Figure \ref{fig:indres}, where the gray and cyan colored tuples are deleted from the original relations and added to the delta relations. 
Note that in the state depicted in Figure \ref{fig:indres}, there are no satisfying assignments to any of the rules in Figure \ref{fig:deltaprog}.
\end{example}

\subsection{Step Semantics}\label{sec:step}
This semantics offers a non-deterministic fine-grain rule activation similar to 
the fact-at-a-time semantics for datalog in the presence of functional dependencies \cite{AbiteboulBD12,AbiteboulDV14}. 
%
We denote the \emph{state} of the database at step $t$ by $D^t = \{R_i^t\}, \Delta^t = \{\Delta_i^t\}, i = 1$ to $m$, and inductively define step semantics as follows.

\begin{definition}\label{def:stepsem}
Let $D$ and $P$ be a database and a delta program over the schema $\mathbf{R} \cup \mathbf{\Delta}$ respectively. 
According to step semantics, at step $t=0$, we have $\Delta_i^{t} = \emptyset$ and $R_i^{t} = $ the relation $R$ in $D$. 
For each step $t>0$, make a non-deterministic choice of an assignment $\alpha: body(r) \to D^t$ to a rule $r\in P$ such that  $head(r) = \Delta_i(\vec{X})$, $tup = \alpha(head(r))$, and update $\Delta_i^{t+1} \gets \Delta_i^t \cup \{tup\}$, and $R_i^{t+1} \gets R_i^{t} \setminus \Delta_i^{t+1}$. For $j \neq i$, $\Delta_j^{t+1} \gets \Delta_j^t$, and $R_j^{t+1} \gets R_j^{t}$. The result of step semantics 
$Step(P,D)$ 
is a minimum size set of non-delta tuples $S$, such that $S = D^0 \setminus D^t$ and $D^{t} = D^{t+1}$. 
\end{definition}

\amir{R2,D6:} {\revb The result of step semantics is then the minimum possible number of tuples that are deleted by a sequence of single rule activations. If there is more than one sequence that results in a minimum number of derived delta tuples, step semantics non-deterministically outputs one of the sets of tuples associated with one of the sequences.}
Step semantics has two uses: (1) simulate a subset of SQL triggers (``delete after delete'') to determine the logic in which they will operate in case there is a need for each trigger to operate separately and immediately update the table from which it deleted a tuple and then evaluate whether another trigger needs to operate (similar to row-by-row semantics, but for multiple triggers), and (2) DC-like semantics can also be simulated with this semantics (see paragraph at the end of this section).

\begin{example}\label{ex:step}
Reconsider the database in Figure \ref{fig:dblp} and the rules in Figure \ref{fig:deltaprog}. \amir{R2,D6:} {\revb We will now demonstrate a sequence of rule activations that results in the smallest set of deleted tuples, which is the result of step semantics.}
\begin{compactenum}
\item At step $t=1$, there is one satisfying assignments to rule (0) that derives $\Delta(\f_2)$. We update $\Delta^1_{\fnd} = \{\f_2\}$, $\fnd^1 = \{\f_1\}$
\item At step $t=2$, there are two satisfying assignments to rule (1). We choose the assignment to rule (1) deriving $\Delta(\au_2)$, and update $D^1$ so it includes the change $\Delta^2_{\auth} = \{\au_2\}$, $\auth^2 = \{\au_1, \au_3\}$
\item In step $t=3$, we have three satisfying assignments: to rules (1), (2), and (3). Suppose we choose the one satisfying rule (1) and derive $\Delta(\au_3)$. $D^2$ is now updated such that $\Delta^3_{\auth} = \{\au_2, \au_3\}$, $\auth^3 = \{\au_1\}$
\item In step $t=4$, there are two assignments to rule (2) and two to rule (3). We choose the one deriving $\Delta(\ww_1)$ and update $D^3$ with $\Delta^4_{\writes} = \{\ww_1\}$, $\writes^4 = \{\ww_2\}$. Note that in the next step, the assignment to rule (3) deriving $\Delta(\pp_1)$ will not be possible due to this update
\item In step $t=5$, there is an assignment to rule (2) and two assignments to rule (3). We choose the one deriving $\Delta(\ww_2)$ and update $D^4$ with $\Delta^5_{\writes} = \{\ww_1, \ww_2\}$, $\writes^5 = \emptyset$. 
\end{compactenum}
The result for this example is depicted in Figure \ref{fig:indres} where the gray and green tuples are deleted from the original relations and added to the delta relations. 

\end{example}

\subsection{Stage Semantics}\label{sec:stage}
Stage semantics separates the evaluation process into stages so at each stage we employ all satisfying assignments to derive \emph{all possible tuples}, and update the delta relations and the original relations (after all possible tuples are found). 
At each stage $t$ of evaluation (similarly to the semin\"aive algorithm \cite{AbiteboulHV95}), we compute all tuples for $\Delta_i$ relations and update the relations $R_i$ in this stage by $R_i^{t} = R_i^{t-1} \setminus \Delta_i^t$. 

\begin{definition}\label{def:stagesem}
Let $D$ and $P$ be a database and a delta program over the schema $\mathbf{R} \cup \mathbf{\Delta}$, respectively. 
According to stage semantics, at stage $t=0$, $\Delta_i^{t} = \emptyset$ and $R_i^{t}$ is the relation $R_i$ in $D$. 
For each stage $t>0$, $\Delta_i^t \gets \Delta_i^{t-1} \cup \{tup ~|~ tup = \alpha(head(r)), r\in P, \alpha[body(r)] \in D^{t-1}, \alpha:body(R)\to D^{t-1}\}$, and $R_i^t \gets R_i^{t-1} \setminus \Delta_i^{t}$. The result of stage semantics, denoted $Stage(P,D)$, is the set of non-delta tuples $S$, such that $S = D^0 \setminus D^{t}$ and $D^{t} = D^{t+1}$. 
\end{definition}

This semantics can be used to simulate a subset of SQL triggers to determine the logic in which they will operate in case there is a need for several stages of deletions of tuples, i.e., the triggers lead to a cascade deletion.

\begin{example}\label{ex:stage}
Reconsider the database in Figure \ref{fig:dblp} and the rules in Figure \ref{fig:deltaprog}. 
Assume we want to perform cascade deletion through triggers such that a deletion of the $\auth$ tuple including the $\fnd$ tuple including $ERC$ will delete its recipients' $\auth$ tuples, and the latter will result in the deletion of the associated $\writes$ and $\pub$ tuples. 
The following describes the operation of stage semantics simulating this process:
\begin{compactenum}
\item At the first stage, there is one assignments to rule (0) deriving $\Delta(\f_2)$, we update $\Delta_{\fnd} = \{\f_2\}$, $\fnd = \{\f_1\}$
\item At the second stage, we use the two assignments to rule (1) to derive $\Delta(\au_2)$ and $\Delta(\au_3)$. We update the database so that $\auth = \{\au_1\}$, $\Delta_{\auth} = \{\au_2,\au_3\}$
\item In the next stage, we use the two assignments to rule (2) and the two assignments to rule (3) to derive $\Delta(\pp_1)$, $\Delta(\pp_2)$, $\Delta(\ww_1)$ and $\Delta(\ww_2)$, and update the database as $\writes = \emptyset$, $\pub = \emptyset$, $\Delta_{\writes} =  \{\ww_1,\ww_2\}$, $\Delta_{\pub} =  \{\pp_1,\pp_2\}$
\end{compactenum}
For any stage $>3$, the state of the database will be identical, so this is the result of stage semantics, shown in Figure \ref{fig:indres} where the tuples in gray, green, and pink are deleted from the original relations and added to the delta relations.
\end{example}

 
Since the delta relations are monotone and can only be as big as the base relation, we can show the following (for brevity, the formal proofs are deferred to the full version).

  \begin{proposition}\label{prop:fixpointstage}
 Let $\mathbf{R}$ be a relational schema. For every database and delta program over $\mathbf{R}$, stage semantics will converge to a unique fixpoint.
 \end{proposition}

\amir{R1,D3:}
 \begin{proof}[Proof Sketch]
 {\reva 
 As stage semantics is rule-order independent and deterministic, at stage $t$ we add all the $\Delta_i$ tuples that can be derived from $D^t$ to get $\Delta_i^{t+1}$, and further delete all the tuples in $\Delta_i^{t+1}$ from $R_i^t$ to get $R_i^{t+1}$. 
 Furthermore, the number of tuples with relations in $\mathbf{R}$ is monotonically decreasing. Thus, there exists a stage in which no more tuples with these relations who satisfy the rules exist. This is the stage that defines the fixpoint.
  }
 \end{proof}

\subsection{End Semantics}\label{sec:end}
Finally, as a baseline, we define end semantics following standard datalog evaluation of delta relations.

\begin{definition}\label{def:endsem}
Let $D$ and $P$ be a database and a delta program over the schema $\mathbf{R}\cup \mathbf{\Delta}$. 
For $t=0$, we have $\Delta_i^{t} = \emptyset$ and $R_i^{t}$ is the relation $R$ in $D$. 
According to end semantics, at each state $t > 0$, $R_i^t \gets R_i^{0}$, and 
$\Delta_i^t \gets \Delta_i^{t-1} \cup \{tup ~|~ tup = \alpha(head(r)), r\in P, \alpha[body(r)] \in D^{t-1}, \alpha:body(R)\to D^{t-1}\}$. Denote the fixpoint of this semantics as $T$, i.e., $D^{T} = D^{T+1}$
At state $T$, $R_i^T \gets R_i^{0} \setminus \Delta_i^{T-1}$, $\Delta_i^T \gets \Delta_i^{T-1}$. The result of end semantics 
$End(P,D)$ 
is the set of non-delta tuples $S = D^0 \setminus D^{T}$.
\end{definition}

This is the standard datalog semantics in the sense that it treats the relations in $\mathbf{\Delta}$ as regular intensional relations and only updates them during the evaluation. Once the evaluation process is completed, the relations in $\mathbf{R}$ are updated. 

\begin{example}\label{ex:end}
For the database and rules in Figures \ref{fig:dblp} 
and \ref{fig:deltaprog}, 
all possible delta tuples will be derived using the rules as shown in Example \ref{ex:fixpont_end}, i.e, $\{\Delta(\f_2), \Delta(\au_2),\Delta(\au_3),\Delta(\ww_1),\Delta(\ww_2),\Delta(\pp_1),$ $\Delta(\pp_2),$ $\Delta(\cc)\}$. 
Then, after the derivation process is done, the tuples $\{\f_2,\au_2,\au_3,\ww_1,\ww_2,\pp_1,\pp_2,\cc\}$ will be deleted, to get the database appearing in Figure \ref{fig:indres} where the gray, green, pink and orange colored tuples are deleted from the original relations and added to the delta relations.
\end{example}

As end semantics is closely related to datalog evaluation, it inherits the basic property of converging to a unique fixpoint.


\subsection{Stabilizing Sets and Problem Statement}\label{sec:stablizing}
After defining delta programs, we introduce the notion of a {\em stable database} with respect to a delta program.

\begin{definition}\label{def:stabilizing}
Given a database $D$ over a schema $\mathbf{R} \cup \mathbf{\Delta}$, and a delta program $P$, $D$ is a stable database w.r.t $P$ if $\{\alpha(head(r)) ~|~ r\in P, \alpha:body(r)\to D, \alpha(body(r))\in D\} = \emptyset$, i.e., $D$ does not satisfy any rule in $P$.
\end{definition}

\begin{example}\label{ex:stable}
Reconsider the database in Figure \ref{fig:dblp} and the rules in Figure \ref{fig:deltaprog}. If we remove the tuples included in the result of end semantics in Example \ref{ex:end} and add their corresponding delta tuples, we would have a stable database.
\end{example}


Alternatively, we can say that a stable database w.r.t. a delta program is a database where no delta tuples can be generated. A database is said to be unstable if it is not stable. 






\begin{definition}\label{def:stabilizingdep}
Given an unstable database $D$ w.r.t a delta program $P$ over a schema $\mathbf{R}\cup  \mathbf{\Delta}$, a stabilizing set for $D$ is a set of tuples $S$ such that $(D\setminus S) \cup \Delta(S)$ is stable. 
\end{definition}

\begin{example}\label{ex:minstab}
Returning to Example \ref{ex:stable}, 
a stabilizing set would be $S = \{\f_2, \au_2, \au_3, \ww_1, \ww_2, \pp_1, \pp_2, \cc\}$, as the database without these tuples and with the tuples in $\Delta(S)$ does not satisfy any of the rules in Figure \ref{fig:deltaprog}.
\end{example}








Our objective is to study the relationship between the results of all semantics we have defined, study the complexity of finding them, and devise efficient algorithms for this purpose.  

\begin{definition}[Problem Definition]\label{def:minstabilizing}
Given $(D, P, \sigma)$, where $D$ is a database and $P$ is a delta program over schema $\mathbf{R}, \mathbf{\Delta}$, and $\sigma$ is a semantics, the desired solution is the result of $\sigma$ w.r.t. $D$ and $P$, denoted by $\sigma(D,P)$.
\end{definition}

\textbf{Initialization of the database and the deletion process.~}
\amir{R2,D14}
{\revb
The deletion process can start in two ways. When the given database contains tuples that violate the constraints expressed by the delta program. This is a popular scenario for data repair. 
Another scenario is where the initial database is stable and the user wants to delete a specific set of tuples. At start, we assume $\Delta_i = \emptyset$ for all $i$.
To start the deletion process, we add a rule for each tuple $R_i(\vec{C})$ of the form $\Delta_i(\vec{C}) \dl R_i(\vec{C})$. 

\begin{example}\label{ex:init}
Consider a slightly different schema than the database in Figure \ref{fig:dblp} where the \pub\ table also mentions the conference in which each paper was published and the delta rule $\Delta_{\pub}(p_1,t_1,conf_1):-\pub(p_1,t_1,conf_1),\pub(p_2,t_1,conf_2)$ stating that no two papers with the same title can be in published in two two different conferences. An unstable database with two tuples $\pub(1,X,C_1)$ and $\pub(1,X,C_2)$ will violate this rule and start the deletion process. 
In our running example, however, we would like to start the deletion process by deleting the tuple $\f_2$, and for this we have defined rule (0) in Figure \ref{fig:deltaprog}.
\end{example}
}



We can observe the following:
\begin{proposition}\label{prop:basic}
Given a database $D$, a delta program $P$, and a semantics $\sigma$, both $D$ and $\sigma(P, D)$ are always stabilizing sets under all four semantics, where $\sigma(P, D)$ is the result of $\sigma$ given $P$ and $D$. In other words, a stabilizing set always exists. 
\end{proposition}

\amir{R2,D11:} {\revb Intuitively, 
if the database is stable, a stabilizing set is the empty set. Otherwise, the entire database is a stabilizing set. 
Additionally, the result of each semantics is defined as the set of non-delta tuples $S$ such that $(D\setminus S) \cup \Delta(S)$ is stable.
Note that sometimes these sets and the results of the different semantics are identical. E.g., if there is only one tuple in the database and one delta rule that deletes it, then this tuple forms the unique stabilizing set and will be returned by 
all semantics. }
Moreover, the results of independent and step semantics may not be unique:


\cut{
\begin{corollary}\label{prop:exists}
Given a delta program $P$ and an unstable database $D$ w.r.t $P$, there exists a stabilizing set.
\end{corollary}
}


\begin{proposition}\label{prop:not_unique}
There exist a database $D$ and a delta program $P$ such that there are two possible results for independent and step semantics.
\end{proposition}
\amir{R2,D12 (fixed):} To see this, consider the database $D = \{R_1(a), R_2(b)\}$ and a program with two rules (1) $\Delta_1(x) :- R_1(x), R_2(y)$, and {\revb (2) $\Delta_2(y) :- R_1(x), R_2(y)$}. 
For independent and step semantics, there are two equivalent solutions: $\{R_1(a)\}$ derived from rule (1), or $\{R_2(b)\}$ derived from rule (2).

\amir{R2,D1:}
{\bf Expressiveness of delta rules.~}
We discuss some forms of constraints that are captured by delta rules.
DCs \cite{ChomickiM05} can be written as a first order logic statement: $\forall \vec{x_1}, \ldots, \vec{x_m} ~ \neg(R_1(\vec{x_1}), \ldots, R_m(\vec{x_m}), \varphi(\vec{x_1}, \ldots, \vec{x_m}))$. 
$\varphi(\vec{x_1}, \ldots, \vec{x_m})$ is a conjunction of atomic formulas of the form $R_i[A_k] \circ R_j[A_l]$, $R_i[A_k] \circ \alpha$, where $\alpha$ is a constant, and $\circ \in \{<, >, =, \neq, \leq, \geq\}$.
Given a DC, $C$, of this form, we translate it to the following delta rule:
\begin{scriptsize}
\begin{lstlisting}
$\Delta_{1}(\vec{x_1}) :- {R_1}(\vec{x_1}), \ldots, {R_m}(\vec{x_m}), \{A^i_k\circ A^j_l ~|~ R_i[A_k] \circ R_j[A_l] \in C\},$ $\{A^i_k\circ \alpha ~|~ R_i[A_k] \circ \alpha \in C\}$
\end{lstlisting}
\end{scriptsize}
The first part of the body contains the atoms used in $C$, the second part contains the comparison between different attributes in $C$ and the third contains the comparison between a attribute and a constant in $C$. 
For independent semantics, the head of the rule can be any delta atom $\Delta_i(\vec{x_i})$. 
$Ind(P,D)$ will then be the smallest set of tuples that should be deleted such that the rule is not satisfied, i.e., from each set of tuples that violate $C$, at least one tuple will be deleted. I.e., $Ind(P,D)$ will be the smallest set of tuples that needs to be deleted such that the database will comply with $C$.
Step semantics can also mimic this by adding a rule for each atom in the rule corresponding to the $C$. We will have $m$ rules and each will have as a head one of the atoms participating in the DC. 
Thus, for each set of of tuples violating $C$, we have a set of $m$ rules allowing us to delete any tuple from this set. Note that in both $Ind(P,D)$ and $Step(P,D)$, only one tuple from the violating set would be deleted. 

Similarly, we can show that delta rules, along with the appropriate semantics, can express Domain Constraints \cite{DeutchF19}, ``after delete, delete'' SQL Triggers \cite{DBSystems0020812}, and Causal Rules without recursion \cite{RoyS14} (whose syntax inspired delta rules).

We now compare the results obtained from the semantics in terms of set containment and size.
\begin{proposition}\label{prop:replationships}
Given a database $D$ and a delta program $P$, we have:
\begin{compactenum}
    \item $|Ind(P,D)| \leq |Step(P,D)|, |Stage(P,D)|$, and there is a case where $|Ind(P,D)| < |Step(P,D)|, |Stage(P,D)|$
     \item $Stage(P,D) \subseteq End(P,D)$, and there is a case where $Stage(P,D) \subsetneq End(P,D)$
     \item $Step(P,D) \subseteq End(P,D)$ , and there is a case where $Step(P,D) \subsetneq End(P,D)$
     \item There exists cases where $Step(P,D) \subsetneq Stage(P,D)$ and cases where $Stage(P,D) \subsetneq Step(P,D)$
 \end{compactenum}
\end{proposition}

\section{Complexity of finding results}\label{sec:complexity}
We now analyze the complexity of finding the result of each semantics. 


\textbf{End semantics.} We follow datalog-like semantics, so the stabilizing set according to end semantics is unique and defined by the single fixpoint. Therefore, we can utilize the standard datalog semantics, treating relations in $\mathbf{\Delta}$ as intensional and deriving all possible delta tuples from the program. After the evaluation is done, we update the relations in $\mathbf{R}$ by removing from them the delta tuples that have been derived.

\cut{
\begin{proposition}\label{prop:endpoly}
 Given a database $D$ and a delta program $P$, Finding $End(P,D)$ is PTime.
 \end{proposition}
 }
 
 \textbf{Stage semantics.~} Similar to end semantics, 
 for stage semantics, if we evaluate the program over the database, we would arrive at a fixpoint. Here, we apply a different evaluation technique, separating the evaluation into stages. At each stage of evaluation, we derive all possible tuples through satisfied rules, and update the database. We continue in this manner until no more tuples can be derived. Therefore, we have the following proposition: 

\begin{proposition}\label{prop:stagepoly}
 Given a database $D$ and a delta program $P$, computing $End(P,D)$ and $Stage(P,D)$ is PTime in data complexity.
 \end{proposition}

\textbf{Independent and step semantics.~} Unlike end and stage semantics, the other two semantics are computationally hard:
 
 \begin{proposition}\label{prop:step-indep-hard}
 Given a delta program $P$, an unstable database $D$ w.r.t $P$, and an integer $k$, it is NP-hard in the value of $k$ to decide whether $|Ind(P,D)| \leq k$ or $|Step(P,D)| \leq k$.
 \end{proposition}
 
 \begin{proof}[proof sketch]
 We reduce the decision problem of minimum vertex cover to finding $Step(P,D)$ and $Ind(P,D)$. Given a graph $G=(V,E)$ and an integer $k$, we define an unstable database $D$: for every $(u,v)\in E(G)$ we have $E(u,v), E(v,u)\in D$ and for every $v\in V(G)$ we have $VC(v) \in D$. 
 For independent semantics we define the delta program: (1) $\Delta_{VC}(x):- E(x,y), VC(x), VC(y)$, (2) $\Delta_{VC}(x):- VC(x), \Delta_E(x,y)$, (3) $\Delta_{VC}(y):-VC(y), \Delta_E(x,y)$. 
For step semantics, we only need rule (1). 
\amir{R2,D15:} {\revb Rules (2) and (3) are only used in the reduction to independent semantics to make the derivation of tuples of the form $E(a,b)$ not worthwhile (as in this semantics, tuples can be removed from $E$ and added to $\Delta_E$ without being derived).}
We can show that a vertex cover of size $\leq k$ is equivalent to $|Ind(P,D)| \leq k$ with the first program and $|Step(P,D)| \leq k$ with the second program (detailed in the full version).
 \end{proof}
 
 
 Naturally, if we consider the search problem, $k$ is unknown and, in the worst case, may be the size of the entire database.

\section{Handling Intractable Cases}\label{sec:algo}
We now present techniques to handle independent and step semantics.

\subsection{Algorithm for Independent Semantics}\label{sec:ind_algo}
Our approach relies on the provenance represented as a Boolean formula \cite{GreenKIT07}, where the provenance of each tuple is a DNF formula, each clause describing a single assignment and delta tuples are negated variables. 



Algorithm \ref{algo:independent} uses this idea to find a stabilizing set. We generate the provenance of each {\em possible} delta tuple (not only ones that can be derived using an operational semantics and the rules) represented as a Boolean formula (line \ref{line:eval}). 
This is a DNF formula for each delta tuple, where tuples with relations in $\mathbf{R}$ are represented as their own literals and tuples in with relations in $\mathbf{\Delta}$ are represented as the negation of their counterpart tuples with relations in $\mathbf{R}$. 
In lines \ref{line:initformula}--\ref{line:formula} we connect these formulae using $\lor$ into one formula representing the provenance of all the delta tuples (this is a disjunction of DNFs). We negate this formula, resulting in a conjunction of CNFs. 
We then find a {\em satisfying assignment giving a minimum number of True values to negated variables}. 
{\revb In the negated formula, each satisfied clause says that at least one of the tuples needed for the assignment the clause represents is not present in the database. An assignment that gives the minimum number of negated variables the value True represents the minimum number of tuples whose deletion from the database and addition of their delta counterparts would stabilize the database. }
Changing negated variables to positive ones and vice-versa will give us an instance of the {\em min-ones SAT} problem \cite{KratschMW16} (line \ref{line:solver}), where the goal is to find a satisfying assignment to a Boolean formula, which maps the minimum number of variables to $True$. 
In line \ref{line:output}, we output the facts whose negated form is mapped to True.


\begin{algorithm}[t]
\begin{footnotesize}
    \SetKwInOut{Input}{Input}\SetKwInOut{Output}{Output}
    \LinesNumbered
    \Input{Delta program $P$, unstable database $D$}
    \Output{A stabilizing set $S\subseteq D$} \BlankLine
    Consider all possible tuples in $t\in D \cup \Delta(D)$ and store the DNF provenance for each tuple $t$\;\label{line:eval}
    Let $F$ be an empty Boolean formula\;\label{line:initformula}
    \ForEach{$t\in P(D)$}
    {\label{l:loop}
        $F\gets F \lor Prov(t)$\;\label{line:formula}
    }
    $\alpha \gets \texttt{Min-Ones-SAT}(\neg F)$\;\label{line:solver}
    output $\{t' ~|~ \alpha (\neg t')=True$\};\label{line:output}
\caption{Find Stabilizing Set - Independent}
\label{algo:independent}
\end{footnotesize}
\end{algorithm} \DecMargin{1em}

\begin{example}
Reconsider the database in Figure \ref{fig:dblp} and the program composed of the rules in Figure \ref{fig:deltaprog}. Algorithm \ref{algo:independent} generates the provenance formula and negates it: 
\begin{scriptsize}
\begin{equation*}
\begin{split}
\neg \f_2\land (\neg \au_2\lor \neg \ad_2\lor \f_2) \land (\neg \au_3\lor \neg \ad_3\lor \f_2) \land (\neg \pp_1\lor \neg \ww_1\lor \au_2) \land\\ (\neg \pp_2\lor \neg \ww_2\lor \au_3) \land (\neg \cc\lor \pp_1\lor \neg \ww_1 \lor \neg \ww_2)
\end{split}
\end{equation*}
\end{scriptsize}
It then generates the assignment giving the value True to the smallest number of {\revb negated} literals in line \ref{line:solver}.
This satisfying assignment is $\alpha$ such that $\alpha(\f_2) = \alpha(\ad_2) = \alpha(\ad_3) = False$ and gives every other variable the value True. Finally, in line \ref{line:output}, the algorithm returns the set of tuples that $\alpha$ mapped to False, i.e., $\{\f_2,\ad_2,\ad_3\}$, as in Example \ref{ex:ind}.
\end{example}

\noindent
\textbf{Correctness:~}
If procedure min-ones SAT finds the minimum satisfying assignment, Algorithm \ref{algo:independent} outputs $Ind(P,D)$. Yet, any satisfying assignment would form a stabilizing set. 


\noindent
\textbf{Complexity:~}
Given a database $D$ and a program $P$, the complexity of computing the provenance Boolean formula is $|D|^{O(|P|)}$; the time to use a solver to find the minimum satisfying assignment. 
Theoretically, such algorithms are not polynomial, however, they are efficient in practice.

\subsection{Algorithm for Step Semantics}\label{sec:step_algo}
\amir{R2,D16,D17:}
We describe a greedy algorithm for step semantics. We will use the concepts of {\em provenance graph} and the {\em benefit} of a tuple. 
{\revb A provenance graph \cite{DeutchMRT14} is a collection of derivation trees \cite{AbiteboulHV95}. A derivation tree of a tuple, $T=(V,E)$, illustrates the tuples that participated in its derivation (the set of nodes $V$), and the process and rules that were used \cite{DeutchGM15} (each rule that uses $t_1, \ldots, t_k$ to derive $t$ is modeled by edges from $t_1, \ldots, t_k$ to $t$). 
When there are several derived tuples of interest, a provenance graph joins together derivation trees, by allowing for the same tuple, used in the derivations of multiple tuples, to be associated with a single node, and allow multiple edges to be adjacent to it. 
When the entire provenance is concerned, the graph shows, for every tuple $t$, all the derivations involving $t$.}
In our case, only delta tuples are derived, so we define the provenance graph as follows: each tuple is associated with a node and there is an edge from $t_1$ to $\Delta(t_2)$ if $t_1$ participates in an assignment resulting in $\Delta(t_2)$. The {\em benefit} of each non-delta node $t$ is the number of assignments it participates in minus the number of assignments $\Delta(t)$ participates in. 
Algorithm \ref{algo:step} stores the provenance as a graph and for each node $t$ we store its benefit $b_t$ (line \ref{l:end}). 

We consider the nodes of the provenance graph $G$ in each layer and the set of assignments $Assign$. 
For each layer $i$ in $G$, denoted by $G_i$, we greedily choose to add to the stabilizing set the tuple $t$, where $\Delta(t)$ is in layer $i$ and $b_t$ is the maximum across all tuples $t$ where $\Delta(t)$ is in layer $i$. 
We then delete the subgraph induced by $\{\Delta(t')~|~ \forall \alpha \in Assign \text{ s.t. } Im(\alpha) = \Delta(t') ~\exists t_k \in Dom(\alpha) \cap S \land t' \neq t_k\}$. 
In words, we delete all delta tuples, such that each one of their assignments contains a tuple $t_k$ that was chosen to be deleted, except $\Delta(t_k)$ itself and the tuples reachable from it. We continue this process until only the delta counterparts of the selected tuples remain in the provenance graph. This ensures that we only delete delta tuples that cannot be generated by any assignment.

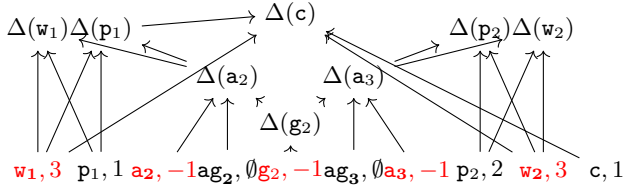
\begin{figure}
\hspace{-3.5mm}
\begin{tikzpicture}[scale=0.42, ]
\node (da1) at (2,1.5) {$\Delta(\f_2)$};
\node (da2) at (0,3) {$\Delta(\au_2)$};
\node (da3) at (4,3) {$\Delta(\au_3)$};
\node (dp1) at (-4,4.5) {$\Delta(\pp_1)$};
\node (dp2) at (8,4.5) {$\Delta(\pp_2)$};
\node (dw1) at (-6,4.5) {$\Delta(\ww_1)$};
\node (dw2) at (10,4.5) {$\Delta(\ww_2)$};
\node (dc) at (2,5) {$\Delta(\cc)$};
 \node[red] (w1) at (-6,0) {$\mathbf{\ww_1},3$};
 \node (p1) at (-4,0) {$\pp_1,1$};
 \node[red] (a2) at (-2,0) {$\mathbf{\au_2},-1$};
 \node (ad1) at (0,0) {$\mathbf{\ad_2},\emptyset$};
 \node[red] (a1) at (2,0) {$\f_2,-1$};
 \node (ad2) at (4,0) {$\mathbf{\ad_3},\emptyset$};
 \node[red] (a3) at (6,0) {$\mathbf{\au_3},-1$};
 \node (p2) at (8,0) {$\pp_2,2$};
 \node[red] (w2) at (10,0) {$\mathbf{\ww_2},3$};
 \node (c) at (12,0) {$\cc,1$};
 \draw[->] (a1) -- (da1);
 
 \draw[->] (da1) -- (da2);
 \draw[->] (ad1) -- (da2);
 \draw[->] (a2) -- (da2);
 
 \draw[->] (da1) -- (da3);
 \draw[->] (ad2) -- (da3);
 \draw[->] (a3) -- (da3);
 
 \draw[->] (da2) -- (dp1);
 \draw[->] (w1) -- (dp1);
 \draw[->] (p1) -- (dp1);
 
 \draw[->] (da2) -- (dw1);
 \draw[->] (w1) -- (dw1);
 \draw[->] (p1) -- (dw1);
 
 \draw[->] (da3) -- (dp2);
 \draw[->] (w2) -- (dp2);
 \draw[->] (p2) -- (dp2);
 
 \draw[->] (da3) -- (dw2);
 \draw[->] (w2) -- (dw2);
 \draw[->] (p2) -- (dw2);
 
 \draw[->] (w1) -- (dc);
 \draw[->] (w2) -- (dc);
 \draw[->] (dp1) -- (dc);
 \draw[->] (c) -- (dc);
\end{tikzpicture}
\caption{Provenance graph for $D$ in Figure \ref{fig:dblp} and the program in Figure \ref{fig:deltaprog}. Red tuples are chosen for the set returned by Algorithm \ref{algo:step}}
\label{fig:provgraph}
\vspace{-2mm}
\end{figure}

\IncMargin{1em}
\begin{algorithm}[t]
\begin{footnotesize}
    \SetKwInOut{Input}{Input}\SetKwInOut{Output}{Output}
    \LinesNumbered
    \Input{Delta program $P$, unstable database $D$}
    \Output{A stabilizing set $S\subseteq D$} \BlankLine
    Store the directed provenance graph $G$ of $End(P, D)$\;\label{l:end}
    Compute $b_t$ for each non-delta tuple $t$\;\label{l:benefit}
    $Assign \gets \{\alpha ~|~ \alpha \text{ is an assignment that derives } \Delta(t) \in \Delta(End(P, D))\}$\;
    $S \gets \emptyset$\;
    \ForEach{Layer $1\leq i \leq L$}
    {\label{l:layer}
        \While{$\exists \Delta(t) \in G_i$ s.t. $t\not\in S$}
        {\label{l:existsdelta}
            $t_m = arg\,max_{t\in V(G), \Delta(t) \in V(G_i)} b_t$\;
            $S \gets S \cup \{t_m\}$\;\label{l:add}
            $G \gets G \setminus G[\Delta(t')~|~ \forall \alpha \in Assign \text{ s.t. } Im(\alpha) = \Delta(t') ~\exists t_k \in Dom(\alpha) \cap S \land t' \neq t_k]$\;\label{l:delete_from_g}
        }
    }
    output $S$;\label{line:output2}
\caption{Find Stabilizing Set - Step}
\label{algo:step}
\end{footnotesize}
\end{algorithm}\DecMargin{1em}

\begin{example}
Reconsider our running example. 
Its provenance graph according to end semantics is shown in Figure \ref{fig:provgraph}. After computing $b_t$ for all the leaf tuples, we begin iterating over the layers of the graph. In layer $1$ we only have $\Delta(\f_2)$, with $b_{\f_2} = -1$, so we choose it. Since $\f_2$ is only connected to $\Delta(\f_2)$, we do not change $G$. We then continue to layer $2$ where we have $\Delta(\au_2)$ and $\Delta(\au_3)$. We arbitrarily choose $\au_2$ as $b_{\au_2} = b_{\au_3} = -1$, and do not change $G$. After that, we choose $\au_3$ and again not change $G$. In layer $3$, we have $\ww_1,\ww_2, \pp_1, \pp_2$ where $ b_{\pp_2} < b_{\pp_1} < b_{\ww_1} = b_{\ww_2}$, so we choose arbitrarily to include $\ww_1$ in $S$. We then delete from $G$ the subgraph induced by $\Delta(\ww_1)$. Since there are more delta tuples in this layer we continue to choose $\ww_2$ and delete from $G$ the subgraph induced by $\Delta(\ww_2)$. 
Since there are no more delta tuples in layers $3$ and $4$ except $\Delta(\ww_1), \Delta(\ww_2)$ where $\ww_1, \ww_2 \in S$, we return $S = \{\f_2,\au_2,\au_3,\ww_1,\ww_2\}$.
\end{example}

\noindent
\textbf{Correctness:~}
Algorithm \ref{algo:step} returns a stabilizing set according to step semantics; the minimum set returned is $Step(P,D)$. 





\noindent
\textbf{Complexity:~}
Given a database $D$ and a program $P$, the overall complexity of Algorithm \ref{algo:step} is $|D|^{O(|P|)}$, since 
it is the size of the provenance graph. 
\section{Implementation \& Experiments}\label{sec:experiments}

We have implemented our algorithms in Python 3.6 with the underlying database stored in PostgreSQL 10.6. Delta rules are implemented as SQL queries and delta relations are auxiliary relations in the database. \amir{R2,D16:} {\revb For Algorithm \ref{algo:independent} we have used the Z3 SMT solver \cite{MouraB08} and specifically, the relevant part that allows for the formulation of optimization problems such as Min-Ones-SAT \cite{BjornerPF15}, which draws on previous work in this field \cite{NieuwenhuisO06,abs-1202-1409,LiAKGC14}.}
For Algorithm \ref{algo:step}, we have used Python's NetworkX package \cite{hagberg2008exploring} to model the graph and manipulate it as required by the algorithm. 
\amir{R2,D18:} {\revb The approach used to evaluate the results of all semantics is a standard n\"aive evaluation, evaluating all rules iteratively, and terminating when no new tuples have been generated.}
The experiments were performed on Windows 10, 64-bit, with 8GB of RAM
and Intel Core Duo i7 2.59 GHz processor, {\revc except for the \holo\ comparison which was performed on Ubuntu 18 on a VMware workstation 12 with 6.5GB RAM allotted. The reason for that is that the Torch package version 1.0.1.post2 required for \holo\ did not run on Windows.}

{\bf Databases: }
We have used a fragment of the MAS database \cite{mas}, containing academic information about universities, authors, and publications. It includes over 124K tuples and the following relations: Organization(\underline{oid}, name), Author(\underline{aid}, name, oid), Writes(aid, pid), Publication(\underline{pid}, title, year), Cite(citing, cited).
 We have also used a fragment of the TPC-H dataset \cite{tpc}, which included 376,175 tuples. This dataset includes 8 tables (customer, supplier, partsupp, part, lineitem, orders, nation, and region).
 
\begin{table}[]
    \centering \tiny
    \caption{MAS Programs \amir{R3,D5,6: said it is too long. Shorten? Also, programs that inspect cascading should have adjacent numbers}}\label{tbl:programs}
    \begin{tabular}{| c | l | }
        \hline {\bf Num.} & \multicolumn{1}{c|}{\bf Program} \\
        \hline
        1 & \specialcell{
            (1) $\Delta_A(aid,n,oid):- A(aid,n,oid), n = C_1$\\
            (2) $\Delta_W(aid,pid):- W(aid,pid), aid = C_2$} \\
        \hline
        2 & \specialcell{
            (1) $\Delta_W(aid,pid):- W(aid,pid), A(aid,n,oid), aid = C$} \\
        \hline
        3 & \specialcell{
            (1) $\Delta_A(aid,n,oid):- W(aid,pid), A(aid,n,oid), aid = C$\\
            (2) $\Delta_W(aid,pid):- W(aid,pid), A(aid,n,oid), aid = C$} \\
         \hline
        4 & \specialcell{
            (1) $\Delta_A(aid,pid):- O(oid,n_2), A(aid,n,oid), oid = C$\\
            (2) $\Delta_O(aid,pid):- O(oid,n_2), A(aid,n,oid), oid = C$} \\
         \hline
         5 & \specialcell{
            (1) $\Delta_A(aid,n,oid):- A(aid,n,oid), n = C_1$\\
            (2) $\Delta_W(aid,pid):- W(aid,pid), \Delta_A(aid,n,oid)$} \\
         \hline
         6 & \specialcell{
            (1) $\Delta_A(aid,n,oid):- A(aid,n,oid), n = C_1$\\
            (2) $\Delta_W(aid,pid):- W(aid,pid), \Delta_A(aid,n,oid)$\\
            (3) $\Delta_P(pid,t):- P(pid,t), \Delta_W(aid,pid), A(aid,n,oid)$} \\
         \hline
         7 & \specialcell{
            (1) $\Delta_P(pid,t):- P(pid,t), pid = C$\\
            (2) $\Delta_C(pid,cited):- C(pid,cited), \Delta_P(pid,t)$\\
            (3) $\Delta_C(citing,pid):- C(citing,pid), \Delta_P(pid,t)$} \\
         \hline
         8 & \specialcell{
            (1) $\Delta_A(aid,n,oid):- W(aid,pid), A(aid,n,oid), aid = C$\\
            (2) $\Delta_W(aid,pid):- W(aid,pid), A(aid,n,oid), aid = C$\\
            (3) $\Delta_P(pid,t):- P(pid,t), \Delta_W(aid,pid), A(aid,n,oid)$\\
            (4) $\Delta_P(pid,t):- P(pid,t), W(aid,pid), \Delta_A(aid,n,oid)$} \\
         \hline
         9 & \specialcell{
            (1) $\Delta_A(aid,n,oid):- A(aid,n,oid), n = C$\\
            (2) $\Delta_W(aid,pid):- W(aid,pid), \Delta_A(aid,n,oid)$\\
            (3) $\Delta_P(pid,t):- P(pid,t), \Delta_W(aid,pid)$\\
            (4) $\Delta_C(pid,cited):- C(pid,cited), \Delta_P(pid,t), pid < C$} \\
         \hline
         10 & \specialcell{
            (1) $\Delta_O(oid,n_2):- O(oid,n_2), oid = C$\\
            (2) $\Delta_A(aid,n,oid):- A(aid,n,oid), \Delta_O(oid,n_2)$\\
            (3) $\Delta_W(aid,pid):- W(aid,pid), \Delta_A(aid,n,oid)$\\
            (4) $\Delta_P(pid,t):- P(pid,t), \Delta_W(aid,pid)$} \\
         \hline
         \hline
        {\revc 11--15} & \specialcell{
            $\Delta_C(pid,c_2):- \{\{\{\{\{C(pid,c_2)\}^{11}, P(t,pid)\}^{12},$\\ $W(aid,pid)\}^{13}, A(aid,n,oid)\}^{14}, O(oid,n_2)\}^{15}$}\\
        \hline \hline
        {\revc 16--20} & \specialcell{
            (1) $\Delta_O(oid,n_2):- O(oid,n_2), oid = C$ (p. 16-20)\\
            (2) $\Delta_A(aid,n,oid):- A(aid,n,oid), \Delta_O(oid,n_2)$ (p. 17-20)\\
            (3) $\Delta_W(aid,pid):- W(aid,pid), \Delta_A(aid,n,oid)$ (p. 18-20)\\
            (4) $\Delta_P(pid,t):- P(pid,t), \Delta_W(aid,pid)$ (p. 16-20)\\
            (5) $\Delta_C(citing,pid):- C(citing,pid), \Delta_P(pid,t)$ (p. 20)}\\
        \hline
    \end{tabular}
\end{table}

\begin{table}[]
    \centering \tiny
    \caption{TPC-H Programs}\label{tbl:tpch_programs}
    \begin{tabular}{| c | l | }
        \hline {\bf Num.} & \multicolumn{1}{c|}{\bf Program} \\
        \hline
         1 & \specialcell{
            (1) $\Delta_{PS}(sk,X):- PS(sk,X), S(sk,Y), sk < C$\\
            (2) $\Delta_{LI}(sk,X):- LI(sk,X), \Delta_{PS}(sk,Y)$
            } \\
         \hline
         2 & \specialcell{
            (1) $\Delta_{PS}(sk,X):- PS(sk,X), sk < C$\\
            (2) $\Delta_{LI}(sk,X):- LI(sk,X), \Delta_{PS}(sk,Y)$} \\
        \hline
        3 & \specialcell{
            (1) $\Delta_{PS}(sk,pk,X):- PS(sk,pk,X), S(sk,Y), P(pk,Y),$\\
                $sk < C$\\
            (2) $\Delta_{LI}(sk,X):- LI(sk,X), \Delta_{PS}(sk,Y)$}\\
        \hline
        4 & \specialcell{
            (1) $\Delta_{LI}(ok,X):- LI(ok,X), ok < C_2$\\
            (2) $\Delta_S(sk,X):- S(sk,X), \Delta_{LI}(sk,ok,Y)$\\
            (3) $\Delta_C(ck,X):- C(ck,X), O(ok,ck,Y), \Delta_{LI}(ok,Z)$} \\
        \hline
        5 & \specialcell{
            (1) $\Delta_N(nk,X):- N(nk,X), nk = C_3$\\
            (2) $\Delta_S(nk,X):- S(nk,X), \Delta_N(nk,Y), C(nk,Z)$\\
            (3) $\Delta_C(nk,X):- S(nk,X), \Delta_N(nk,Y), C(nk,Z)$}\\
        \hline
        6 & \specialcell{
            (1) $\Delta_O(ck,X):- O(ck,X), C(ck,Y), ck < C_4$\\
            (2) $\Delta_{PS}(sk,X):- PS(sk,X), S(sk,Y), sk < C_4$\\
            (3) $\Delta_{LI}(sk,X):- LI(ok,X), \Delta_O(ok,Y)$\\
            (4) $\Delta_{LI}(sk,X):- LI(sk,X), \Delta_{PS}(sk,Y)$\\}\\
        \hline
    \end{tabular}
    \vspace{-2mm}
\end{table}

{\bf Test programs:} 
{\revc Tables \ref{tbl:programs} and \ref{tbl:tpch_programs} show the programs we have used for the MAS and TPC-H datasets experiments, respectively. 
We use the first letter of each table as an abbreviation, and denote by $C/C_i$ a constant we have assigned to an attribute. 
The programs were designed for different scenarios to compare the four semantics and highlight the manner in which each semantics is advantageous. 
The programs can roughly be divided into three sets: (1) those that are meant to mimic the semantics of integrity constraints such as DCs (programs 1--4, 11--15 in Table \ref{tbl:programs}), (2) those that are meant to perform cascade deletion (programs 5, 9, 10, and 16--20 in Table \ref{tbl:programs} and programs 1--3 in Table \ref{tbl:tpch_programs}), and (3) those that mix between the two (programs 6--8 in Table \ref{tbl:programs}, and programs 4--6 in Table \ref{tbl:tpch_programs}). 
For programs that express integrity constraints, independent semantics would guarantee a minimum size repair while the other semantics may delete a larger number of tuples. For example, in program 2 in Table \ref{tbl:programs}, using end, stage or step semantics may yield a result composed of Writes tuples which will likely not be minimal in size. If instead we use independent semantics, we could have a result of a single \auth\ tuple.
For programs that are purely designed for cascade deletion, we expect the result of all semantics to be the same and therefore the fastest and most accurate algorithm should be used, i.e., end or stage semantics. 
For the programs that perform a mix of the two options, it would depend on the desired result. For example, program 8 in Table \ref{tbl:programs} is designed to distinguish between stage and step semantics, where stage semantics will not be able to use rules 3 and 4, while step semantics will not be able to derive all delta tuples from both rules 1 and 2. }

\textbf{Setting and highlights:~}
We have focused on four different aspects in our experimental study: (1) the relationship between the sets found for each semantics; (2) the size of the result set computed by each algorithm; (3) the algorithms execution times and their breakdown and (4) a comparison of our approach with PostgrSQL and MySQL triggers, {\revc and a comparison with the state-of-the-art data repair system \holo\ \cite{RekatsinasCIR17} that repairs cells instead of deleting tuples.}
We have manually checked that Algorithms \ref{algo:independent}, \ref{algo:step} output the actual result for programs 1, 2, 3, 5--9 (where the sizes of the result are small enough to be manually verified). Hence, we refer to the output given by these algorithms as the result of the two semantics. 
\amir{R2,D18: highlights}
{\revb All of the algorithms computed the results in feasible time (the average runtimes for end, stage, step and independent were 16.9, 21.1, 389.5, and 73 seconds resp. for the programs in Table \ref{tbl:programs}). 
In general, computing the results of end and stage semantics is faster than those of step and stage semantics. Thus, for programs that perform cascade deletion (e.g., 16--20 in Table \ref{tbl:programs}), where the result for all semantics is the same, it may be preferable to use end or stage semantics. 
For programs such as 11--15 in Table \ref{tbl:programs}, where there is a clear difference between the results, users may choose the desired semantics they wish to enforce, while aware of the difference in performance. As an example, for these programs, independent semantics would correspond to the semantics of DCs (but would be slower to compute the repair), while the other ones would correspond to triggers. 
We also demonstrate the discrepancy between the results of the different semantics using specific programs and Table \ref{tbl:mss_containment} showing the relationships between the results. 
For example, for program 8 in Table \ref{tbl:programs}, there is a no containment of the result of stage in the result of step semantics and vice versa. }

\begin{figure*}[!htb]
\captionsetup{justification=centering}
    \centering
    \begin{subfigure}[b]{0.33\textwidth}
        \centering
        \centering
        \includegraphics[trim=0cm 0.5cm 0cm 0cm, width=2in]{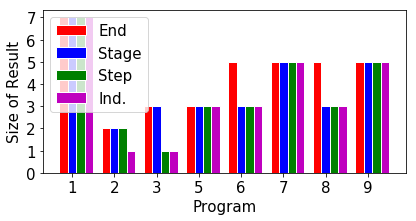}
        \caption{Size of results Programs 1--10}
        \label{graph:prog_size}
    \end{subfigure}%
    \begin{subfigure}[b]{0.33\textwidth}
        \centering
        \includegraphics[trim=0cm 0.5cm 0cm 0cm, width=2.1in]{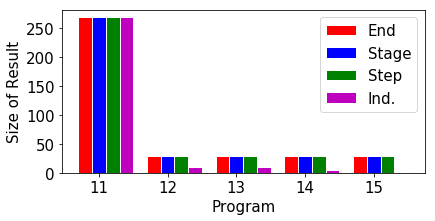}
        \caption{Size of results; Programs 11--15}
        \label{graph:join_size}
    \end{subfigure}%
    \begin{subfigure}[b]{0.33\textwidth}
        \centering
        \includegraphics[trim=0cm 0.5cm 0cm 0cm, width=2.2in]{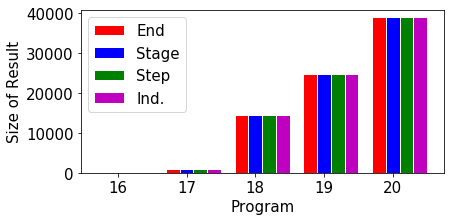}
        \caption{Size of results; Programs 16--20}
        \label{graph:num_rules_size}
    \end{subfigure}
    \caption{Comparison of result sizes for the four semantics with the programs from Table \ref{tbl:programs}}
    \label{fig:sizes}
    \vspace{-3mm}
\end{figure*}
\begin{figure*}[!htb]
\captionsetup{justification=centering}
    \centering
    \centering
        \includegraphics[trim=0cm 0cm 0cm 0cm, width=6.5in]{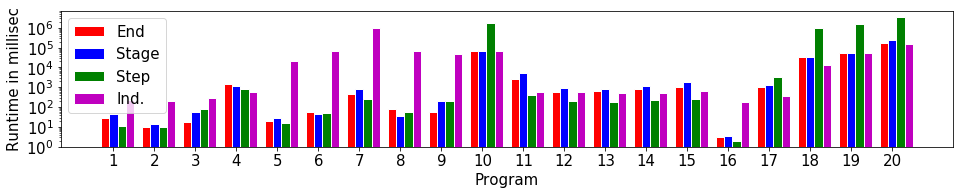}
    \caption{{\revb Execution time for finding the results of the four semantics with the programs from Table \ref{tbl:programs}} \amir{R2,D18: change to base 10}}
    \label{fig:runtimes}
    \vspace{-3mm}
\end{figure*}

\begin{figure*}[!htb]
\captionsetup{justification=centering}
    \centering
    \begin{subfigure}[b]{0.25\textwidth}
        \centering
        \includegraphics[trim=0cm 0.5cm 0cm 0cm, width=1.6in]{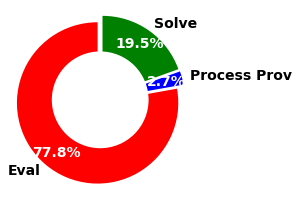}
        \caption{Algo. \ref{algo:independent} (1--15)}
        \label{graph:breakdown_ind}
    \end{subfigure}%
    \begin{subfigure}[b]{0.25\textwidth}
        \centering
        \includegraphics[trim=0cm 0.5cm 0cm 0cm, width=1.6in]{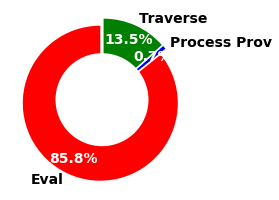}
        \caption{Algo. \ref{algo:step}  (1--15)}
        \label{graph:breakdown_step}
    \end{subfigure}%
    \begin{subfigure}[b]{0.25\textwidth}
        \centering
        \includegraphics[trim=0cm 0.5cm 0cm 0cm, width=1.6in]{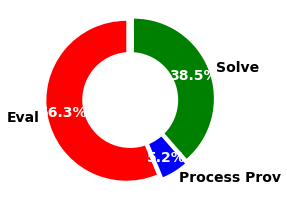}
        \caption{Algo. \ref{algo:independent} (16--20)}
        \label{graph:breakdown_ind_num_rules}
    \end{subfigure}%
    \begin{subfigure}[b]{0.25\textwidth}
        \centering
        \includegraphics[trim=0cm 0.5cm 0cm 0cm, width=1.8in]{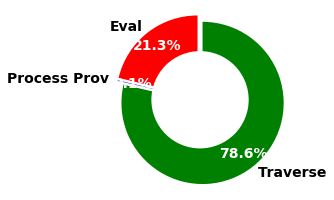}
        \caption{Algo. \ref{algo:step}  (16--20)}
        \label{graph:breakdown_step_num_rules}
    \end{subfigure}
    \caption{Runtime breakdown for programs 1--15 and 16--20, and Algorithms \ref{algo:independent} (ind. sem.) and \ref{algo:step} (step sem.)}
    \label{fig:breakdown}
    \vspace{-3mm}
\end{figure*}

\begin{figure}[!htb]
    \centering
    \begin{subfigure}[b]{0.49\linewidth}
        \centering
        \includegraphics[trim=0cm 0.5cm 0cm 0cm, width=1.6in]{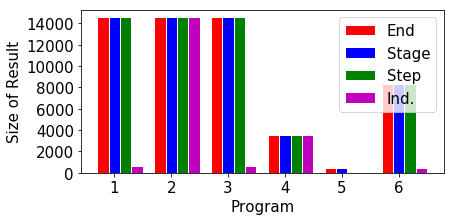}
        \caption{Size of results; TPC-H}
        \label{graph:tpch_size}
    \end{subfigure}%
    \begin{subfigure}[b]{0.49\linewidth}
        \centering
        \includegraphics[trim=0cm 0.5cm 0cm 0cm, width=1.6in]{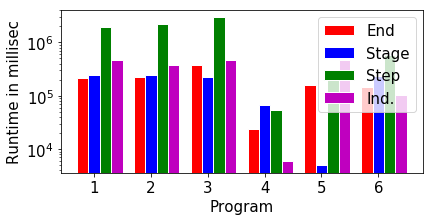}
        \caption{{\revb Runtime; TPC-H}}
        \label{graph:tpch_runtime}
    \end{subfigure}
    \caption{Comparison of results sizes and runtimes for the four semantics with TPC-H programs}
    \label{fig:tpch_results}
    \vspace{-4mm}
\end{figure}

\begin{table}[]
    \centering \tiny
    \caption{Containment of results for the programs in Tables \ref{tbl:programs} and \ref{tbl:tpch_programs}}\label{tbl:mss_containment}
    \begin{tabular}{| c | c | c | c |}
        \hline {\bf Program} & $\mathbf{Step = Stage}$ & $\mathbf{Ind \subseteq Stage}$ & $\mathbf{Ind \subseteq Step}$ \\
        \hline
        1  & \cellcolor{green}\cmark & \cellcolor{green}\cmark & \cellcolor{green}\cmark \\
        \hline 2 & \cellcolor{green}\cmark & \cellcolor{red}\xmark & \cellcolor{red}\xmark \\ 
        \hline 3  & \cellcolor{red}\xmark & \cellcolor{green}\cmark & \cellcolor{green}\cmark \\
        \hline 4  & \cellcolor{red}\xmark & \cellcolor{green}\cmark & \cellcolor{green}\cmark \\
        \hline 5  & \cellcolor{green}\cmark & \cellcolor{green}\cmark & \cellcolor{green}\cmark \\
        \hline 6 & \cellcolor{green}\cmark & \cellcolor{green}\cmark & \cellcolor{green}\cmark \\
        \hline 7 & \cellcolor{green}\cmark & \cellcolor{green}\cmark & \cellcolor{green}\cmark \\
        \hline 8  & \cellcolor{red}\xmark & \cellcolor{red}\xmark & \cellcolor{green}\cmark \\
        \hline 9 & \cellcolor{green}\cmark & \cellcolor{green}\cmark & \cellcolor{green}\cmark \\
        \hline 10  & \cellcolor{green}\cmark & \cellcolor{green}\cmark & \cellcolor{green}\cmark \\ 
        \hline \hline
        11 & \cellcolor{green}\cmark & \cellcolor{green}\cmark & \cellcolor{green}\cmark \\
        \hline 12--15 & \cellcolor{green}\cmark & \cellcolor{red}\xmark & \cellcolor{red}\xmark \\ 
        \hline  \hline
        16--20 & \cellcolor{green}\cmark & \cellcolor{green}\cmark & \cellcolor{green}\cmark \\
        \hline \hline
        T-1  & \cellcolor{green}\cmark & \cellcolor{red}\xmark & \cellcolor{red}\xmark \\
        \hline 
        T-2 & \cellcolor{green}\cmark & \cellcolor{green}\cmark & \cellcolor{green}\cmark \\ 
        \hline 
        T-3  &  \cellcolor{green}\cmark & \cellcolor{red}\xmark & \cellcolor{red}\xmark \\
        \hline 
        T-4  &  \cellcolor{green}\cmark & \cellcolor{red}\xmark & \cellcolor{red}\xmark \\
        \hline 
        T-5  & \cellcolor{red}\xmark & \cellcolor{green}\cmark & \cellcolor{green}\cmark \\
        \hline 
        T-6 & \cellcolor{green}\cmark & \cellcolor{red}\xmark & \cellcolor{red}\xmark \\
        \hline
    \end{tabular}
\end{table}

\textbf{Containment of results:~}
Table \ref{tbl:mss_containment} shows the relationship between the results generated for the different semantics. The table has three columns: $Step = Stage$, describing whether the result of stage semantics is equal to the result of step semantics, $Ind \subseteq Stage$ and $Ind \subseteq Step$ which capture whether the result of independent semantics is contained in the result of stage and step semantics respectively. The other relationships always hold, as shown in Figure \ref{fig:mss_classes}. 
We start by reviewing the results for the programs in Table \ref{tbl:programs}. For program 2, there is no containment of the result of independent semantics, since it includes a single Author tuple which cannot be derived, so it cannot be in the results of stage or step semantics. Programs 3 and 4 are composed of two rules with the same body, so the result of stage semantics contains all derivable tuples while the result of step and independent semantics contains only one Author tuple (this is also evident in Figure \ref{graph:prog_size} for program 3). 
Program 8 was designed based on the proof of Proposition \ref{prop:replationships}, and thus ``separates'' between step and stage semantics. For programs 12--15, the tuples chosen for the result in independent semantics cannot be derived and hence there is no containment. Finally, for programs 16--20, all derived tuples have to be included in the result, according to all semantics and, therefore, all the conditions in the table are true. 
The results for the TPC-H programs in Table \ref{tbl:tpch_programs} are shown in the lower part of Table \ref{tbl:mss_containment} with the prefix ``T''. 
As for the first column, we found that only for program 5, $Stage \not\subseteq Step$. This program contains two rules with the same body, and step semantics was able to delete fewer tuples by selecting the minimal set of Customer and Supplier delta tuples to derive. 
For the second and third columns, the result of independent semantics was not contained in the result of either step or stage or both for all programs except programs 2 and 5, as Algorithm \ref{algo:independent} deleted tuples that were not derivable by other semantics. 

\textbf{Results size:~}
Figure \ref{fig:sizes} depicts the results size for the different programs in Table \ref{tbl:programs}. For the chart in Figure \ref{graph:prog_size}, we included all programs except for 4 and 10, as they would have distorted the scale. For program 4, the sizes were 956 for end and stage semantics and 1 for step and independent semantics. For program 10, the sizes of all results were 24,798. In Figure \ref{graph:prog_size}, as predicted in Figure \ref{fig:mss_classes}, the size of the result of end semantics is always larger than the sizes of the results according to the other semantics. For program 2, the result of independent semantics can be of size 1 (the Author tuple with $aid = C$), whereas all other semantics may include only Writes tuples, since Author tuples cannot be derived. 
Furthermore, note that programs 3 and 4 was designed to have only one tuple in the result of step and independent semantics (the Organization tuple with oid $C$), and all Author tuples along with the Organization tuple for end and stage semantics. 
Figure \ref{graph:join_size} shows the results for programs 11--15. Note that the results of all semantics except for independent semantics can only include Cite tuples. Thus, the results size according to end, stage and step semantics is identical for all programs, but the result size for the independent semantics actually decreases as the number of joins increases. 
In Figure \ref{graph:num_rules_size}, all results sizes are equal for every program since all possible tuples need to be included in the stabilizing set by all semantics. The maximum result size for program 20 was 38,954. 
Figure \ref{graph:tpch_size} shows the sizes of the results for the TPC-H programs in Table \ref{tbl:tpch_programs}, the largest being 14,550 tuples for programs 1, 2, 3 through end, stage and step semantics. 
The rational for the results here is similar, where for programs 1, 3, 5 and 6 Algorithm \ref{algo:independent} (ind. semantics) outputted a smaller result by choosing tuples that were not derived by the rules.

\textbf{Execution times:~}
We have examined the execution time for the algorithms of the four semantics and all programs in Table \ref{tbl:programs} (Figure \ref{fig:runtimes}) and Table \ref{tbl:tpch_programs} (Figure \ref{graph:tpch_runtime}). 
The recorded times are presented in $\log$ scale. When the execution time is not negligible, Algorithms \ref{algo:independent} and \ref{algo:step} require the largest execution time for most programs due to the overhead of generating the Boolean formula and finding the minimum satisfying assignment or generating the provenance graph and traversing it. 
For programs 10 and 16--20, all derived tuples participate in the result of each semantics and, hence, all algorithms have to ``work hard''. In particular, Algorithm \ref{algo:step} has to traverse a provenance graph of 5 layers for program 20. 
The results for Programs 11--15 (single rule with an increasing number of joins) were all fast (the slowest time was 5.5 seconds,  incurred for stage semantics). Thus, an increase in the number of joins does not necessarily reflect an increase in execution time. 
Most computations were dominated either by Algorithm \ref{algo:independent} or \ref{algo:step} as both are algorithms that store and process the provenance as opposed to the two other algorithms for end and stage semantics. 
\amir{R2,D18:} {\revb In some cases, Algorithm \ref{algo:step} is faster than the algorithms for stage and end semantics. This happens when the runtimes are either very small (e.g., programs 1 and 2), or for programs 11--15. In the latter, stage and end semantics have to delete all tuples that are derived through the rule and add their delta counterparts to the database throughout the evaluation process. For Algorithm \ref{algo:step}, after creating the graph, we need to traverse a single layer. 
 }


\textbf{Runtime breakdown for Algorithms \ref{algo:independent} and \ref{algo:step}:~}
Figure \ref{fig:breakdown} shows the breakdown of the execution time for both algorithms. 
We have computed the average distribution of execution time across programs 1--15 and programs 16--20 in Table \ref{tbl:programs}.
In Figure \ref{graph:breakdown_ind}, most of the computation time is devoted to the evaluation and storage of the provenance (Eval). 
The second most expensive phase is finding the minimum satisfying assignment for the Boolean formula in the SAT solver (Solve). Converting the provenance to a Boolean formula does not require much time (Process Prov). Similarly, for Algorithm \ref{algo:step} in Figure \ref{graph:breakdown_step}, most of the time is spent on evaluation and provenance storing (Eval). Traversing and choosing the nodes with maximum benefit is the second most expensive phase (Traverse) and finally, converting the provenance into a graph and determining the benefits is negligible (Process Prov). 
Figure \ref{graph:breakdown_ind_num_rules} shows the breakdown for programs 16--20. Algorithm \ref{algo:independent} devotes a larger percentage to solving the Boolean formula.
Figure \ref{graph:breakdown_step_num_rules} shows that most of the execution time is devoted to traversing the provenance graph and finding the tuples to include in the outputted set.


\textbf{Comparison with Triggers:~}
Triggers \cite{triggers,melton2001sql,DBSystems0020812} is a standard approach for updating the database when constraints are violated. We have implemented Programs 3, 4, 5, 8 and 20 from Table \ref{tbl:programs} using triggers both in PostgreSQL and in MySQL. For programs 3 and 4, where two triggers are programmed to fire at the same event, the PostgreSQL triggers were fired alphabetically by their assigned name while the MySQL triggers fired by the order in which they were written. Due to this fact, for program 4, the PostgreSQL triggers deleted all Author tuples associated with a single organization, instead of one Organization tuple. 
In these scenarios, using step semantics would have yielded a smaller result. 
Both PostgreSQL and MySQL triggers have led to the same result as the four semantics for program 5. 
For program 8, PostgreSQL triggers, the Writes tuples were deleted using the trigger version of rule 2 and then the Publication tuples were deleted using the trigger version of rule 4. For the MySQL implementation, the results depended on the order in which the triggers were written. When the Author triggers were written before the Writes triggers, the tuples with this relation were deleted, and then their associated Publication tuples. When the order was reversed, the Writes tuples were deleted and and then their associated Publication tuples. 
If we would have applied stage semantics instead, only the Author and Writes tuples would have been deleted. Using step semantics, we would have only deleted an Author tuple and the Publication tuples associated with it (regardless of the name of the trigger or the order in which it was written). 
For program 20, the same number of tuples were deleted by the PostgreSQL triggers as for the four semantics (shown in Figure \ref{graph:num_rules_size}). The MySQL triggers were not able to terminate computation before the connection to the server was lost. 
Computing the trigger results for programs 3, 4, and 8 was negligible in terms of execution time for both PostgreSQL and MySQL implementations. For program 20, it was 3.3 minutes for PostgreSQL triggers as opposed to 2.9 minutes for end semantics, and 4.25 minutes for stage semantics, 40.3 minutes for step semantics, and 2.4 minutes for the independent semantics.


\begin{table}[]
\caption{{\revc Number of over deletions ($+$) for each of the four semantics compared with number of under repaired tuples ($-$) by \holo\ for an increasing the number of errors. Note that in contrast to \holo\ all of our semantics always fixed all violations}}\label{tbl:holo_comp}
\centering
\begin{tiny}
\begin{tabular}{c | >{\bf}c >{\bf}c >{\bf}c >{\bf}c | c}
& \multicolumn{4}{c|}{Deleted Tuples} & {Repaired Tuples} \\
\hline
Errors & Ind & Step & Stage & End & \holo \\
\hline \hline
100  & +0 & +0 & +389 & +389 & -26 \\
200  & +0 & +1 & +479 & +479 &  -60 \\
300  & +0 & +5 & +630 & +630 &  -128 \\
500  & +0 & +16 & +786 & +786 &  -234 \\
700  & +0 & +21 & +878 & +878 &  -480 \\
1000 & +0 & +34 & +1000 & +1000 &  -693 \\
\hline
\end{tabular}
\end{tiny}
\end{table}

\setlength{\tabcolsep}{3pt}
\begin{table}[]
\caption{{\revc Number of tuples that violate a DC with other tuples in the table  after/before the repair for both \holo\ and our four semantics. Some tuples participate in multiple violations. 
}}\label{tbl:holo_violations}
\begin{tiny}
\begin{tabular}{c | c c c c c | >{\bf}c}
& \multicolumn{5}{c|}{\holo} & {Semantics} \\
\hline
Errors & $DC_1$ & $DC_2$ & $DC_3$ & $DC_4$ & Total & Total \\
\hline \hline
100  & 22/42 & 30/46 & 0/112 & 0/415 & 52/615 & 0/615 \\
200  & 42/82 & 78/110 & 0/208 & 0/563 & 120/963 & 0/963 \\
300  & 94/158 & 98/140 & 64/302 & 187/761 & 443/1361 & 0/1361 \\
500  & 134/254 & 116/246 & 218/500 & 464/1015 & 932/2015 & 0/2015 \\
700  & 198/320 & 182/364 & 580/716 & 872/1272 & 1832/2672 & 0/2672 \\
1000 & 238/474 & 186/520 & 962/1006 & 1355/1612 & 2741/3612 & 0/3612 \\
\hline
\end{tabular}
\end{tiny}
\end{table}


\begin{figure}[!htb]
    \centering
    \begin{subfigure}[b]{0.49\linewidth}
        \centering
        \includegraphics[trim=0cm 0.5cm 0cm 0cm, width=1.6in]{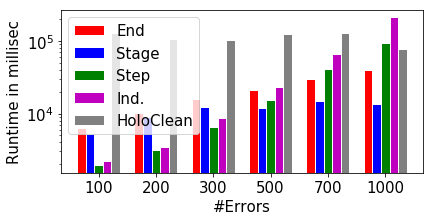}
        \caption{Increasing \#errors}
        \label{graph:holo_runtime_errors}
    \end{subfigure}%
    \begin{subfigure}[b]{0.49\linewidth}
        \centering
        \includegraphics[trim=0cm 0.5cm 0cm 0cm, width=1.6in]{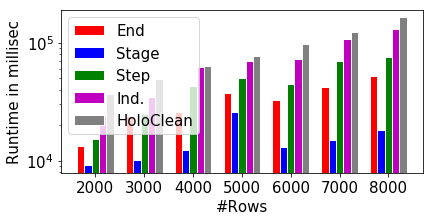}
        \caption{Increasing \#rows}
        \label{graph:holo_runtime_rows}
    \end{subfigure}
    \caption{{\revc Runtime comparison with \holo\ for increasing number of errors (rows set to 5000) and number of rows (errors set to 700)}}
    \label{fig:holo_runtimes}
    \vspace{-3mm}
\end{figure}

\textbf{Comparison with HoloClean:~} 
{\revc
\holo\ \cite{RekatsinasCIR17} is a data repair system that relaxes hard constraints (as opposed to our system that views the delta rules as hard constraints) and uses a probabilistic model to infer cell repairs (instead of tuple deletions) in order to clean the database. It leverages DCs, among other methods, to detect and repair cells. 
\holo\ uses the context of the cell and statistical correlations to repair cells, rather than delete tuples solely based on constraints, as we do for our semantics. In addition, \holo\ does not support cascade deletion. Nevertheless, we have examined what would happen if \holo\ was used in the same context as our system and what would be the difference in results, while also examining the performance of our algorithms for the different semantics in this scenario. 
We have used the code of the system from \cite{holocode}  with the default configuration that allows for a single table to be inputted. 
Our comparison used the \auth\ table as presented at the start of this section with an extra attribute stating the organization name: \auth(\underline{aid}, name, oid, organization).
We have used four DCs, expressed here as delta rules: 

\begin{scriptsize}
\begin{lstlisting}
($DC_1$) $\Delta_{A_1}(a_1,n_1,o_1,on_1)$ :- $A_1(a_1,n_1,o_1,on_1)$, $A_2(a_2,n_2,o_2,on_2),$ $a_1 = a_2, o_1 \neq o_2$
($DC_2$) $\Delta_{A_1}(a_1,n_1,o_1,on_1)$ :-  $A_1(a_1,n_1,o_1,on_1),$ $A_2(a_2,n_2,o_2,on_2),$ $a_1 = a_2, n_1 \neq n_2$
($DC_3$) $\Delta_{A_1}(a_1,n_1,o_1,on_1)$ :- $A_1(a_1,n_1,o_1,on_1),$ $A_2(a_2,n_2,o_2,on_2),$ $a_1 = a_2, on_1 \neq on_2$
($DC_4$) $\Delta_{A_1}(a_1,n_1,o_1,on_1)$ :- $A_1(a_1,n_1,o_1,on_1),$ $A_2(a_2,n_2,o_2,on_2),$ $o_1 = o_2, on_1 \neq on_2$
\end{lstlisting}
\end{scriptsize}
Note that these delta rules simulate DCs semantics. E.g., 
the first DC says that there cannot be two tuples with the same $aid$ and a different $oid$ attribute. Thus, if there is such a pair of tuples, the delta rule will delete at least one of them. 
For these DCs, the results of independent and step semantics should be exact in theory (although our algorithms are heuristic so their output may not be identical to the theoretical results), while the results of end and stage semantics should delete all tuples that satisfy any of these constraints. 
For Tables \ref{tbl:holo_comp} and \ref{tbl:holo_violations} we have taken a table of 5000 rows and increased the number of errors.
Table \ref{tbl:holo_comp} shows the results for the number of tuples deleted beyond the minimum required number by each of our semantics and the difference between the number of repairs to cells made by \holo\ (this is identical to the number of repaired tuples) to the number of required repairs. Algorithm \ref{algo:independent} deleted the same number of tuples as the number of errors. The algorithms for the rest of the semantics `over deleted', while \holo\ has performed fewer repairs than needed\footnote{This
is based on the report automatically generated by the system.} outputting an unstable database. 
In Table \ref{tbl:holo_violations} we have measured the number of tuples that violate each DC with another tuple after/before the repair for \holo, where the ``Total'' column shows the sum of violations (the sum may be larger than the size of the set, since tuples may participate in violations through multiple DCs). 
The numbers are the sizes of the results generated by running each DC as an SQL query before and after the repair. 
As guaranteed by Proposition \ref{prop:basic} and by our algorithms, every semantics repairs the database so that there are no sets of tuples that violate a DCs, where \holo\ may leave some violating sets of tuples after the repair. 
Figures \ref{graph:holo_runtime_errors} and \ref{graph:holo_runtime_rows} show the runtime performance for all semantics alongside the performance of \holo\ for an increasing number of errors with 5000 rows and for an increasing number of rows with 700 errors. 
End and stage semantics were faster than the rest, while Algorithms \ref{algo:independent} and \ref{algo:step} had similar performance to that of \holo. 
}

\section{Related Work}\label{sec:related}

\textbf{Data repair.~} \amir{R3W2,3: Add a more stark difference between prev. works and ours} 
{\revc Multiple papers have used database constraints as a tool for fixing (in our terms stabilizing) the database \cite{Afrati2009,RekatsinasCIR17,ChuIP13,BertossiKL13,FaginKK15}. 
The literature on data repair can be divided by two main 
criteria: the types of constraints considered  and 
the methods to repair the database. A wide variety of constraints with different forms and functions
have been proposed. 
Examples include functional dependencies and versions thereof \cite{bohannon2007conditional,koudas2009metric}, and  denial constraints \cite{ChomickiM05}. 
As we have discussed in Section \ref{sec:stablizing}, our model can express various forms of constraints, but our semantics allow these constraints to be interpreted in different ways and not operate according to one specific algorithm or approach. 
Regarding methods of data repair, previous works have considered two main 
approaches: (1) repairing attribute values in cells \cite{RekatsinasCIR17,ChuIP13,KolahiL09,BertossiKL13, LivshitsKR18} and (2) tuple deletion \cite{ChomickiM05,LopatenkoB07,LivshitsKR18}; 
our work focuses on the latter. 
A major advantage of our approach is the ability to perform cascade deletions over multiple relations in the database while following different well-defined semantics (and the admin may choose which one to follow based on the application scenario).
Similar to our independent semantics, a common objective for data repairs  
is to change the database in the minimal way that will make it consistent with the constraints \cite{Afrati2009,FaginKK15, LivshitsKR18}. 
In some scenarios 
a good repair 
can be obtained by changing values in the database and the metric of minimal changes may not work well \cite{RekatsinasCIR17}. However, in our approach as in \cite{ChomickiM05}, we assume that the starting database is \emph{complete}, so the only way to fix it is by deleting tuples and thus we use the minimum cardinality metric to achieve a repair following the delta program; extending delta rules to updates of values is an interesting future work.}
Similar to our declarative repair framework by delta rules, \emph{declarative data repair} 
has been explored from multiple angles  \cite{GalhardasFSSS01,ProkoshynaSCMS15,VolkovsCSM14,HeVSLMPT16,FaginKRV16}; e.g., \cite{FaginKRV16} has focused on the rule-based framework of information extraction from text and includes a mechanism for prioritized DC repairs, while \cite{RekatsinasCIR17} expresses constraints in DDlog \cite{ShinWWSZR15}.

\textbf{Causality in databases.~}
This subject has been the explored in many previous works \cite{MeliouGMS10,MeliouGHKMS10,MeliouRS14}.
Works such as \cite{RoyS14,RoyOS15} consider causal dependencies for explaining results of aggregate queries, 
that start their operation when there is an initial event of tuple deletion called ``intervention'' and repair the database 
if a constraint is violated. In particular, \cite{RoyS14} focused on repairs with respect to foreign keys in  both ways (similar to rules (2) and (3) in Example~\ref{ex:first}), whereas our delta programs can capture these as well as more complex cascaded deletion rules. 
Moreover, interventions can also be applied in our framework, as we can add auxiliary rules to the program that will start the deletion process.

\textbf{Stable model.~}
{\revb Stable model semantics \cite{GelfondL88,DantsinEGV01} is a way of defining the semantics of the answer set of logic programs with negation. 
Stable models use the concept of a {\em reduct} of the program w.r.t. the database instance to define the model.
In stable models, if a tuple does not exist in the database, it means that its negation exists. 
In our model, a tuple that does not exist in $R_i$ does not have to be present in $\Delta_i$, i.e., $\Delta_i$ is not the negation of $R_i$, but is a record of deleted tuples from $R_i$. 
Also in our model, the head atom in each rule can only be a delta atom, rather than a positive atom as in stable model. 
Another relevant work related to our framework is \cite{GatterbauerS10}, where the authors used the concept of stable models to solve the data conflict problem with trust mappings. The way one's belief is updated from others' beliefs is expressed by weighted update rules that are similar in spirit to our delta rules. 
However, in \cite{GatterbauerS10} rules have priorities and the results of the semantics can be computed in PTime under the {\em skeptic} paradigm, while in our framework, delta rules do not have priorities and computing the results of some of our semantics is NP-hard, and they have different usages. 
}

\textbf{Deletion propagation.~}
Classic deletion propagation is the problem of evaluating the effect of deleting tuples from the database $D$ on the view obtained from evaluating a query $Q$ over $D$ \cite{GreenKIT07,GreenKT07,DeutchMRT14}. A more closely related variation is the source side-effect problem  \cite{BunemanKT02,CongFG06,CongFGLL12}, which focuses on finding the minimum set of {\em source tuples in $D$} that has to be deleted for a given tuple $t\in Q(D)$ to be removed from the result. Our approach may be combined with this problem by including the delta program as another input and solving the source side-effect problem given the delta program and a particular semantics. 

\section{Conclusions and Future Work}\label{sec:conc}
In this paper, we presented, for the first time to our knowledge,
a unified framework for repairs that involves deletions. 
We have devised a model to accommodate the constraints and four semantics that capture  behaviors inspired by DCs, a subset of SQL triggers, and causal dependencies, allowing for different interpretation of the same set of constraints. 
We studied the relationships between the results of these semantics, and explored the complexity facet of all four semantics, showing algorithms to solve the tractable cases and heuristics to handle the intractable cases. 
We also describe an extensive experimental evaluation of our algorithms. 

We have focused on programs that are not inherently recursive (Section \ref{sec:prelim}). However, all definitions and results in Sections \ref{sec:prelim}, \ref{sec:deltarules}, and \ref{sec:complexity} also apply to recursive programs. The limitation lies in Algorithms \ref{algo:independent} and \ref{algo:step}, 
that rely on the size of the provenance and its structure.
When the program is inherently recursive, the provenance size may be super-polynomial in the database size. 
Extending our solutions to general recursion is left for future work.

\paragraph{{\bf Acknowledgements}}
\small{
This research has been funded by the European
Research Council (ERC) under the European Union's Horizon 2020
research and innovation programme (Grant agreement No. 804302),
the Israeli Science Foundation (ISF) Grant No. 978/17, NSF awards IIS-1552538, IIS-1703431, NIH award 1R01EB025021-01, and the Google
Ph.D. Fellowship. The contribution of Amir Gilad is part of his Ph.D. research conducted at Tel Aviv University.
}

\bibliographystyle{abbrv}
\bibliography{bibtex.bib}
\clearpage
\section*{Appendix - Full Proofs}

\subsection*{Proofs from Section \ref{sec:deltarules}}

\begin{proof}[Proof of Proposition \ref{prop:fixpointstage}]
 Let $t$ be a stage in the evaluation. It is sufficient to show that the there is only one outcome of the evaluation at this stage. In stage $t$, we evaluate all rules of the delta program over the current database $D^t$. As stage semantics is rule-order independent and deterministic, at stage $t$ we add all the $\Delta_i$ tuples that can be derived from $D^t$ to get $\Delta_i^{t+1}$, and further delete all the tuples in $\Delta_i^{t+1}$ from $R_i^t$ to get $R_i^{t+1}$. 
 Thus, there is just a single option to move to stage $t+1$. As this holds for every stage, we just need to show that the semantics in monotone. 
 Since only delta tuples can be derived and regular tuples are being deleted at the end of every stage, the number of tuples with relations in $\mathbf{R}$ is monotonically decreasing. Every delta rule with relation $\Delta_i$ at its head, has the atom with relation $R_i$ in its body. As the tuples with these relations are exactly the tuples being deleted from the database at each stage, there exists a stage in which no more tuples with these relations who satisfy the rules exist. This is the stage which defines the fixpoint.
 \end{proof}

 \begin{proof}[proof of Proposition \ref{prop:basic}]
It is clear that $D$ is a stabilizing set as no tuples can be derived from an empty database. 

$\sigma(P, D)$ is also a stabilizing set since it results in a database which does not satisfy any rule in $P$ by the Definition of each semantic (Definitions \ref{def:indsem}, \ref{def:stepsem}, \ref{def:stagesem}, \ref{def:endsem}).
\end{proof}

\begin{proof}[proof of Proposition \ref{prop:not_unique}]
Consider the database $D = \{R_1(a), R_2(b)\}$ and the program two rules (1) $\Delta_1(x) :- R_1(x),$ $R_2(y)$, and (2) $\Delta_2(x) :- R_1(x), R_2(y)$. 
Based on all independent and step semantics' definitions, the result can be $\{R_1(a)\}$ derived from rule (1), or $\{R_2(b)\}$ derived from rule (2).
\end{proof}

\subsection*{Proof for Proposition \ref{prop:replationships} (Section \ref{sec:deltarules})}

\begin{proof}[Proof of item 1 in Proposition \ref{prop:replationships}]
First, we show \\$|Ind(P,D)| \leq |\sigma(P,D)|$. Notice that $\sigma(P,D)$ is also a stabilizing set (by Definition \ref{def:stabilizing}). Thus, $Ind(P,D)$ can be equal to $\sigma(P,D)$, or it can be smaller, since a tuple $t\in Ind(P,D)$ is not constrained by the need to be derived by a rule of $P$.

For the second part of the claim, define the database $D = \{R_1(a_1), \ldots, R_1(a_n), R_2(b)\}$ and the program with the single rule $\Delta_1(x) :- R_1(x), R_2(y)$, then the result of independent semantics is $\{R_2(b)\}$, while based on every other semantics, the result is $\{R_1(a_i)~|~ 1\leq i\leq n\}$ since there is no way of deriving $\Delta_2(b)$.
\end{proof}

\begin{proof}[Proof of item 2 in Proposition \ref{prop:replationships}]
For the first part, every delta tuple derived with stage would be derived by end semantics, since at every stage of stage semantics, all tuples that can participate in a derivation of a delta tuple can also participate in the same derivation according to end semantics (as it only updates the database at the end of the evaluation process). 

For 2, consider the database $D = \{R_1(a), R_2(a), R_3(b_1), \ldots,$ $R_3(b_n)\}$ and the following delta program:
\begin{lstlisting}
        (1) $\Delta_1(x)$ :-  $R_1(x)$
        (2) $\Delta_2(x)$ :-  $\Delta_1(x),R_2(x)$
        (3) $\Delta_3(y)$ :- $R_1(x), \Delta_2(x), R_3(y)$
\end{lstlisting}
According to end semantics, we first compute all tuples for delta tuples over $\mathbf{\Delta}$ and update the relations of $\mathbf{R}$ once we have reached the final state $T$ of evaluation. 
Therefore, we add $\Delta_{1}(a)$ (as only rule (1) is satisfied initially), then adding $\Delta_{2}(a)$ (rule (2)) and finally adding $\Delta_{3}(b_i)$ for all $1\leq i\leq n$ (rule (3)). When we have the instance $\{R_1(a), R_2(a), R_3(b_1),$ $\ldots, R_3(b_n), \Delta_{1}(a),$ $\Delta_{2}(a), \Delta_{3}(b_1),\ldots, \Delta_{3}(b_n)\}$, we update the database to $\{\Delta_{1}(a),$ $\Delta_{2}(a),\Delta_{3}(b_1),\ldots, \Delta_{3}(b_n)\}$ and this is the fixpoint database for End semantics. 

Conversely, in stage semantics, we update the relations in $\mathbf{R}$ at every stage, where a stage is defined by the state where there are no more new derivable delta tuples at state $t$. 
In this example, at the first stage, we can derive only $\Delta_{1}(a)$ from rule (1), so we update the database to be $D' = \{R_2(a), \Delta_1(a), R_3(b_1), \ldots, R_3(b_n)\}$ (deleting $R_1(a)$ and adding $\Delta_1(a)$ instead). Then, we derive $\Delta_{2}(a)$ from rule (2) and update the database to be $D'' = \{\Delta_2(a), \Delta_1(a), R_3(b_1), \ldots,$ $R_3(b_n)\}$. $D''$ is the fixpoint database for stage semantics since we cannot derive any new tuple from it.
\end{proof}

\begin{proof}[Proof of item 3 in Proposition \ref{prop:replationships}]
 The same proof for the first part and the same program and database for the second part.
\end{proof}

\begin{proof}[Proof of item 4 in Proposition \ref{prop:replationships}]
 For part 1, consider the database $D = \{R_1(a), R_2(b_1), \ldots, R_2(b_n)\}$ along with the two delta rules:
 \begin{lstlisting}
            (1) $\Delta_1(x)$ :- $R_1(x), R_2(y)$
            (2) $\Delta_2(y)$ :- $R_1(x), R_2(y)$
\end{lstlisting}
With stage semantics, the two rules are satisfied at the initial stage so we have the result which is simply $D$. For step semantics, however, we can fire rule 1 first and update the database so rule 2 cannot be satisfied to get the result $\{R_1(a)\}$.

For part 2, consider the database $D = \{R_1(a), R_2(b), R_3(c_1),\\ \ldots, R_3(c_n)\}$ along with the four delta rules:
 \begin{lstlisting}
            (1) $\Delta_1(x)$ :- $R_1(x), R_2(y)$
            (2) $\Delta_2(x)$ :- $R_1(x), R_2(y)$
            (3) $\Delta_3(z)$ :- $R_3(z), \Delta_1(x), R_2(y)$
            (4) $\Delta_3(z)$ :- $R_3(z), R_1(x), \Delta_2(y)$
\end{lstlisting}
With Stage semantics, rules (1), (2) are satisfied at the initial stage, so after this stage we have $Stage(P,D) = \{R_1(a), R_2(b)\}$. Since in the second stage no no tuples can be derived, this is $Stage(P,D)$. 
For Step semantics, however, we fire rule 1 (or rule 2) first and update the database so rule 2 (1, respectively) cannot be satisfied. We get the initial set $\{R_1(a)\}$ ($\{R_2(b)\}$). Then, rule 3 (or rule 4, respectively) are satisfied by the assignment $R_3(c_i), \Delta_1(a), R_2(b)$ (or $R_3(c_i), R_1(a), \Delta_2(b)$) so we have to add every tuple of the form $R_3(c_i)$ to the result, which means that $Step(P,D) = \{R_1(a), R_3(c_1), \ldots, R_3(c_n)\}$ (or $Step(P,D) = \{R_2(b), R_3(c_1), \ldots, R_3(c_n)\}$). 
\end{proof}

\subsection*{Proof of Proposition \ref{prop:step-indep-hard} (Section \ref{sec:complexity})}

\begin{proof}[Proof of Proposition \ref{prop:step-indep-hard} for independent semantics]
The decision problem is formulated as: given a delta program program $P$, an unstable database $D$ and an integer $k$, does $D$ has a stabilizing set of size at most $k$? 

First we show membership in NP: given a set of tuples $S\subseteq D$ such that $|S| \leq k$, we can verify that $(D\setminus S) \cup \Delta(S)$ does not satisfy any rule of $P$ in polynomial time. 

We prove hardness by showing a reduction from the Vertex Cover decision problem. 
Given a graph $G=(V,E)$ and an integer $k$, we define an unstable database $D$: for every $(u,v)\in E(G)$ we have $E(u,v), E(v,u)\in D$ and for every $v\in V(G)$ we have $VC(v) \in D$. We further define the delta program $P$ with three rules (note that the head of the rules does not matter for independent semantics and is just there for uniformity):
\begin{small}
\begin{lstlisting}
        (1) $\Delta_{VC}(x):- E(x,y), VC(x), VC(y)$
        (2) $\Delta_{VC}(x):- VC(x), \Delta_E(x,y)$
        (3) $\Delta_{VC}(y):-VC(y), \Delta_E(x,y)$
\end{lstlisting}
\end{small}
Observe that $D$ is unstable, as it satisfies rule (1) and that 
this reduction can be done in polynomial time. 
We now prove that $G$ has an vertex cover of size at most $k$ if and only if $D$ has the results of independent semantics is of size at most $k$.

Suppose $Ind(P,D) = S$ such that $|S| \leq k$. 

We first show that $S$ contains only tuples of relation $VC$ and no tuples of relation $E$. 
Assume by contradiction that $S$ contains a tuple $E(a,b)$. 
Therefore, the database $D' = (D\setminus S) \cup \Delta(S)$ would contain the tuple $\Delta_E(a,b)$. There are two options in this case: rules (2) and (3) are satisfied and $D'$ is unstable -- a contradiction to the definition of independent semantics and Proposition \ref{prop:basic}, or $S$ also has to contain $VC(a), VC(b)$, but then $S$ could have been made smaller by containing just $VC(a)$ or $VC(b)$ or both and then all three rules would not have been satisfied -- a contradiction to the definition of independent semantics. 
Thus, $S$ cannot contain a tuple of the form $E(a,b)$.

Denote by $C$ all the vertices $v\in V(G)$ if and only if $VC(v) \in S$. 
Let us now show that $C$ is a vertex cover. Denote $D' = D\setminus S$. 
Let $(u,v)\in E(G)$. We need to show that either $u \in C$ or $v \in C$ or both. Equivalently, we can show that $VC(u)$ and $VC(v)$ are not in $D'$. $D'$ is stable so it does not satisfy the single rule in the program. Therefore, we cannot have the both the facts $VC(v)$ and $VC(u)$ in $D'$ (along with $E(u,v), E(v,u)$). Hence, at least one of them has been deleted. 

Suppose we have an vertex cover $C$ for $G$ such that $|C| \leq k$. We show that $|Ind(P,D)| \leq |\{VC(v) ~|~ v\in C\}|$. We remove the tuple $VC(v)$ from $D$ and add $\Delta_{VC}(v)$ if and only if $v\in C$ (that is, $\{VC(v) ~|~ v\in C\}$ is a stabilizing set) to get the database $D'$. Thus, $D'$ contains only tuples of the form $VC(v)$ for vertices $v\not\in C$. 
We now show that $D'$ is stable, i.e., that $D'$ does not satisfy the single rule of the program. In other words, there are no tuples $VC(v)$, $VC(u)$, $E(u,v)\in D'$. Equivalently, we can show that for every $(u,v) \in E(G)$, either $u\in C$ or $v\in C$ or both. But this is true by definition, so we are done.
\end{proof}

\begin{proof}[Proof of Proposition \ref{prop:step-indep-hard} for step semantics]
The proof is very similar to the proof of Theorem  \ref{prop:step-indep-hard} with a reduction from minimum Vertex Cover. We consider the same database and rule 1 (we refer to it as our program $P$):
\begin{small}
\begin{lstlisting}
        (1) $\Delta_{VC}(x):- E(x,y), VC(x), VC(y)$
\end{lstlisting}
\end{small}
We now prove that $G$ has an vertex cover of size at most $k$ if and only if $|Step(P,D)| \leq k$ w.r.t step semantics.

Suppose $Step(P,D) = S$ such that $|S| \leq k$. 
Denote by $C$ all the vertices $v\in V(G)$ if and only if $VC(v) \in S$. 
Let us now show that $C$ is a vertex cover. Denote $D' = D\setminus S \cup \Delta(S)$. 
Let $(u,v)\in E(G)$. We need to show that either $u \in C$ or $v \in C$ or both. Equivalently, we can show that $VC(u)$ and $VC(v)$ are not in $D'$. $D'$ is stable so it does not satisfy the single rule in the program. Therefore, we cannot have the both the facts $VC(v)$ and $VC(u)$ in $D'$ (along with $E(u,v), E(v,u)$). Hence, at least one of them has been deleted. 

Suppose we have a vertex cover $C$ for $G$ such that $|C| \leq k$. We show that $Step(P,D) \subseteq S$ such that $S = \{VC(v) ~|~ v\in C\}$, so we need to show that $S$ can be derived using step semantics.

\begin{lemma*}
$S$ can be derived using step semantics.
\end{lemma*}

\begin{proof}
For every edge $(a,b)\in E$ for which only $a \in C$, we can derive $\Delta_{VC}(a)$ as we have the tuples $E(a,b), VC(a), VC(b)$ in $D$. Consider an edge $(a,b)$ for which $a,b\in C$. We need to specify a way to derive both under step semantics. We first use the assignment $E(a,b), VC(a), VC(b)$ to derive one of them, say $\Delta_{VC}(a)$. To derive the other vertex, say $\Delta_{VC}(b)$, we choose another edge that is adjacent to $b$, say $(b,c)$, to derive $\Delta_{VC}(b)$ with the assignment $E(b,c), VC(b), VC(c)$. 
Why is there always such an edge $(b,c)$? Since otherwise, $C$ could have been made smaller since $C\setminus \{b\}$ would also be a vertex cover.
\end{proof}

Denote $D' = (D\setminus S) \cup \Delta(S)$. Thus, $D'$ contains only tuples of the form $VC(v)$ for vertices $v\not\in C$. 
We now show that $D'$ is stable, i.e., that $D'$ does not satisfy the single rule of the program. In other words, there are no tuples $VC(v)$, $VC(u)$, $E(u,v)\in D'$. Equivalently, we can show that for every $(u,v) \in E(G)$, either $u\in C$ or $v\in C$ or both. But this is true by definition, so we are done.
\end{proof}

\subsection*{Proofs of Algorithms' Correctness (Section \ref{sec:algo})}

\begin{proof}[Proof of correctness for Algorithm \ref{algo:independent}]
\amir{verify this:}
Denote by $S$ the set returned by Algorithm \ref{algo:independent}. 
Let $F$ denote the Boolean formula of the provenance as defined in the algorithm. $F$ is a disjunction of DNF formulae, so $\neg F$ is a conjunction of CNF formulae, i.e., a CNF formula. Denote $\neg F = C_1\land \ldots \land C_m$ where each $C_i$ represents one assignment that derives a delta tuple. Let $C_i$ be a clause in $\neg F$. We show that the assignment represented by $C_i$ is voided after removing the set defined in line \ref{line:output} of Algorithm \ref{algo:independent}. 
Recall that we initialize all $\Delta_i = \emptyset$. Thus there is at least one clause in $\neg F$ composed entirely of negated variables, so at least one negated variable has to be assigned to True.
In the satisfying assignment $\alpha$ to $\neg F$, there is a literal in $C_i$, $l$, such that $\alpha(l) = True$. If $l = \neg a$, in $F$, the tuple $a$ was used for the assignment $C_i$ and it is added to the $S$ and deleted from $D$ so $C_i$ is voided and cannot be used. Otherwise $l = a$, so in $F$, $\Delta(a)$ was used for the assignment $C_i$. If $\alpha(a) = True$ in $\neg F$, then $\alpha(\neg a) = False$, meaning $a$ is not needed for the satisfying assignment to $\neg F$, and thus $\Delta(a)$ does not appear in $(D\setminus S) \cup \Delta(S)$, and thus $C_i$ is voided in this case as well.
We have established that a satisfying assignment to $\neg F$ correspond to a set of tuples $S\subseteq D$ such that $(D\setminus S) \cup \Delta(S)$ is stable. 
A satisfying assignment to $\neg F$ that gives the minimum number of negated literals the value $True$ will thus result in a set $S$ of minimum size. 



\end{proof}

\begin{proof}[Proof of correctness for Algorithm \ref{algo:step}]
We show two things: (1) Algorithm \ref{algo:step} returns a stabilizing set, and (2) every tuple in this set can be derived according to step semantics.

For (1), the end state of Algorithm \ref{algo:step} is that the provenance graph shows only derivations of tuples included in the stabilizing set $S$. Suppose there is a tuple $\Delta(t)$ that can be derived with the database $(D\setminus S) \cup \Delta(S)$. 
If $\Delta(t)$ is in the graph, it means that there is an assignment $\alpha$ and a set of non-delta tuples $t_1, \ldots, t_k$ that are used to derive it ($k \geq 1$ since $t$ is always used to derive $\Delta(t)$), such that no tuple $t_i$ is in $S$. 
If $t \in S$, then $\Delta(t)$ has been derived and $t$ has been deleted from the database. 
Therefore, $t \not\in S$. This is a contradiction as the only delta tuples that remain in the graph are ones that have their original counterparts in $S$. 


For (2), Let $t\in S$. Suppose  $\Delta(t)$ is derived in layer $j$, then all tuples in the assignment deriving $\Delta(t)$ are in layers $1$ to $j-1$. Furthermore, they are still in the provenance graph in iteration $j$ of the loop in line \ref{l:layer} since otherwise, $\Delta(t)$ would have been deleted itself. 
Thus, $\Delta(t)$ can be derived using step semantics.
\end{proof}

\end{document}